\documentclass[amsthm]{elsart}
\usepackage{yjsco}
\usepackage{natbib}

\usepackage{amsmath}
\usepackage{amssymb}
\usepackage{flushend}
\usepackage{array,multirow,graphicx}
\usepackage{tabularx}
\usepackage{url}
\usepackage{tikz}
\usepackage{manfnt}
\usetikzlibrary{backgrounds,fit,shapes.geometric}
\tikzset{>=stealth}

\usepackage{pgfplots}
\usepackage{subfig}

\newcommand{\sn}{\smallskip\noindent}

\newcommand{\bbbb}{\ensuremath{\mathrm{I\!B}}}
\newcommand{\bbbf}{\ensuremath{\mathrm{I\!F}}}
\newcommand{\bfunc}{\ensuremath{\mathcal{B}}}

\newtheorem{definition}{Definition}
\newtheorem{lemma}{Lemma}
\newtheorem{example}{Example}
\newtheorem{remark}{Remark}
\newtheorem{theorem}{Theorem}

\newcommand{\eqdef}{\stackrel{\text{\normalfont def}}{=}}
\newcommand{\onset}{\ensuremath{\operatorname{on}}}
\newcommand{\offset}{\ensuremath{\operatorname{off}}}
\newcommand{\nat}{\ensuremath{\operatorname{nat}}}

\newcolumntype{L}{>{$}l<{$}}
\newcolumntype{C}{>{$}c<{$}}

\tikzset{%
  terminal vertex/.style={draw,rectangle,inner sep=0pt,minimum width=10pt,minimum height=10pt},
}
\def\low{\mathop{\rm low}}
\def\high{\mathop{\rm high}}

\def\bddFalse{\tikz[baseline=(x.base)] \node[terminal vertex] (x) {$\bot$};}
\def\bddTrue{\tikz[baseline=(x.base)] \node[terminal vertex] (x) {$\top$};}

\newcommand{\T}{\ensuremath{\mathrm{T}}}

\def\xskip{\hskip 7pt plus 3pt minus 4pt}

\newdimen\algindent
\newif\ifitempar \itempartrue 
\def\algindentset#1{\setbox0\hbox{{\bf #1.\kern.25em}}\algindent=\wd0\relax}
\def\algbegin #1 #2{\algindentset{#21}\alg #1 #2} 
\def\aalgbegin #1 #2{\algindentset{#211}\alg #1 #2} 
\def\alg#1(#2). {\medbreak 
  \noindent{\bf#1}{({\it#2\/})}.\xskip\ignorespaces}
\def\algstep#1.{\ifitempar\smallskip\noindent\else\itempartrue
  \hskip-\parindent\fi
  \hbox to\algindent{\bf\hfil #1.\kern.25em}%
  \hangindent=\algindent\hangafter=1\ignorespaces}
\def\slug{\hbox{\kern1.5pt\vrule width2.5pt height6pt depth1.5pt\kern1.5pt}}

\def\hang{\hangindent19pt}
\def\d@anger{\medbreak\begingroup\clubpenalty=10000
 \def\par{\endgraf\endgroup\medbreak} \noindent\hang\hangafter=-2
 \hbox to0pt{\hskip-\hangindent\dbend\hfill}\small}
\outer\def\danger{\d@anger}

\journal{Journal of Symbolic Computation}

\begin{document}

\begin{frontmatter}
\title{Ancilla-free synthesis of large reversible functions
       using binary decision diagrams}
\author[uni,dfki]{Mathias Soeken}
\author[uni]{Laura Tague}
\author[unb]{Gerhard W.~Dueck}
\author[uni,dfki]{Rolf Drechsler}
\address[uni]{Department of Mathematics and Computer Science, University of Bremen, Bremen, Germany}
\address[dfki]{Cyber-Physical Systems, DFKI GmbH, Bremen, Germany}
\address[unb]{Faculty of Computer Science, University of New Brunswick, Fredericton, Canada}

\begin{abstract}
  The synthesis of reversible functions has been an intensively studied research
  area in the last decade.  Since almost all proposed approaches rely on
  representations of exponential size (such as truth tables and permutations),
  they cannot be applied efficiently to reversible functions with more than~15
  variables.

  In this paper, we propose an ancilla-free synthesis approach based on Young
  subgroups using symbolic function representations that can efficiently be
  implemented with binary decision diagrams~(BDDs).  As a result, the algorithm
  not only allows to synthesize large reversible functions without adding extra
  lines, called ancilla, but also leads to significantly smaller circuits
  compared to existing approaches.
\end{abstract}

\end{frontmatter}

\section{Introduction and Background}
Given a bijective function~$f:\bbbb^n\to\bbbb^n$, also called a reversible
function, \emph{synthesis} describes the problem of determining a circuit
composed of reversible gates that realizes~$f$.  If this circuit consists of
exactly~$n$ signal lines, the synthesis is called \emph{ancilla-free}.  In the
last decade, ancilla-free synthesis approaches have been presented that start
from a reversible function represented e.g.~as truth tables~\citep{MMD:2003},
permutations~\citep{SPMH:2003}, or Reed-Muller spectras~\citep{MDM:2007}.  Since
all these representations are always of exponential size with respect to~$n$,
the respective algorithms do not scale well and are thus not efficiently
applicable to large reversible functions.  Reversible functions and circuits
play an important role in quantum computing~\citep{SM:2013} and low-power
computing~\citep{Landauer:1961,Vos:2010,BAP+:2012}.

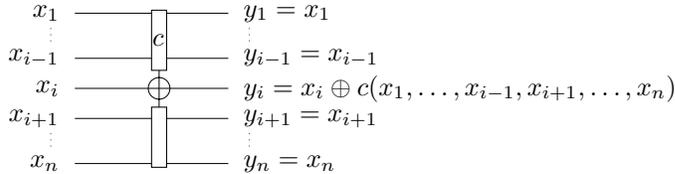
\begin{figure}[t]
  \centering
\begin{tikzpicture}[scale=.8]
\draw[line width=0.300000] (0.500000,2.850000) -- (3.050000,2.85000);
\draw (0.400000,2.850000) node [left] {$x_1$};
\draw (3.1500000,2.850000) node [right] {$y_1 = x_1$};
\draw[line width=0.300000,dotted] (0.100000,2.60000) -- (0.100000,2.30000);
\draw[line width=0.300000,dotted] (3.400000,2.60000) -- (3.400000,2.30000); 
\draw[line width=0.300000] (0.500000,2.100000) -- (3.05000,2.100000);
\draw (0.400000,2.100000) node [left] {$x_{i-1}$};
\draw (3.150000,2.100000) node [right] {$y_{i-1} = x_{i-1}$};
\draw[line width=0.300000] (0.500000,1.600000) -- (3.05000,1.600000);
\draw (0.400000,1.600000) node [left] {$x_i$};
\draw (3.150000,1.600000) node [right] {$y_i = x_i \oplus c(x_1, \dots{}, x_{i-1}, x_{i+1},\dots{}, x_n  ) $};
\draw[line width=0.300000] (0.500000,1.10000) -- (3.05000,1.10000);
\draw (0.400000,1.100000) node [left] {$x_{i+1}$};
\draw (3.150000,1.100000) node [right] {$y_{i+1}=x_{i+1}$};
\draw[line width=0.300000] (0.500000,0.350000) -- (3.05000,0.350000);
\draw (0.400000,0.350000) node [left] {$x_{n}$};
\draw (3.150000,0.350000) node [right] {$y_{n}=x_{n}$};
\draw[line width=0.300000,dotted] (0.100000,0.85000) -- (0.100000,0.55000);
\draw[line width=0.300000,dotted] (3.400000,0.85000) -- (3.400000,0.55000); 
\draw[fill=white,opacity=.98] (1.775,1.9) rectangle ++(0.25,1.0);
\draw (2.135000,2.400000) node [left] {$c$};
\draw[line width=0.3] (1.9,1.6) circle (0.18) (1.90,1.9) -- (1.90,.93);
\draw[fill=white,opacity=.98] (1.775,0.288) rectangle ++(0.25,1.0); 
\end{tikzpicture}
  \caption{Single-target gate}
  \label{fig:single-target-gate}
\end{figure}

One of these truth table based algorithms was presented in~\citep{VR:2008}
and uses \emph{single-target gates} as gate library.  A single-target
gate~$\T[c,i]:\bbbb^n\to\bbbb^n$ with~$\T[c,i](x_1,\dots, x_n)=(y_1,\dots,y_n)$
updates the value of the input variable~$x_i$ with respect to a Boolean
\emph{control function}~$c:\bbbb^{n-1}\to\bbbb$ that is defined on all other
variables.  That is, the target gate computes a new value at the \emph{target
  output} $y_i=x_i\oplus c(x_1,\dots,x_{i-1},x_{i+1},\dots,x_n)$ and leaves all
other output variables unaltered, i.e.~$y_j=x_j$ for~$j\neq i$.
Figure~\ref{fig:single-target-gate} shows the diagrammatic representation of a
single-target gate based on Feynman's notation~(\citeyear{Feynman:1985}).

The algorithm described in~\citep{VR:2008} makes use of the property that each
reversible function~$f:\bbbb^n\to\bbbb^n$ can be decomposed into three
$\bbbb^n\to\bbbb^n$ reversible functions
\begin{equation}
  \label{eq:decomposition}
  f=\T[l,i]\circ f'\circ \T[r,i]
\end{equation}
such that~$f'$ does not change in~$x_i$.  Based on the truth table
representation of~$f$, the algorithm determines two control functions~$l$
and~$r$ from which~$f'=T[l,i]\circ f\circ T[r,i]$ can be determined, since
single-target gates are self-inverse.  Applying the decomposition
in~\eqref{eq:decomposition} for each variable results in~$2n-1$ single-target
gates\footnote{For the last variable only one single-target gate is
  required. The middle part of~\eqref{eq:decomposition} degenerates to the
  identity.}  which composed as a circuit realize~$f$.  The algorithm always
traverses the whole truth table and is therefore exponential with respect to the
number of variables.  Consequently, it cannot efficiently be applied to large
functions.

\sn
The contributions described in this paper are as follows:
\begin{itemize}
\item We propose an algorithm based on a symbolic function representation to
  determine the control functions~$l$ and~$r$ from~\eqref{eq:decomposition}.
  The algorithm makes use of Boolean operations that can efficiently be
  implemented using binary decision diagrams~(BDDs).  Since BDDs allow for a
  more compact representation of many practical functions, the efficiency of the
  proposed algorithm is increased and can therefore be applied to larger
  functions.
\item In~\citep{VR:2008} only one variable ordering for the sequential
  application of~\eqref{eq:decomposition} has been considered, i.e.~the natural
  ordering $x_1,x_2,\dots,x_n$.  We investigate different heuristics to find
  variable orderings that allow for more compact synthesis results.
\item The original algorithm is not applicable to partial functions,
  i.e.~functions for which not all input/output mappings are specified.
  Our algorithm can be adjusted in order to support partial functions.
\item Finally, we provide open source implementations for both the original
  truth table based approach and our proposed approach.
\end{itemize}

\sn The majority of the proposed synthesis algorithms use Toffoli
gates~\citep{Toffoli:1980} as underlying gate library.  Our proposed method uses
single-target gates which can be transformed into a cascade of Toffoli gates.
However, since the lower bound for required Toffoli gates in a reversible
circuit is exponential with respect to the number of
lines~\citep{MDM:05,SAD:14}, the use of Toffoli gates in large reversible
circuits may not be convenient.  In contrast, the upper bound for required
single-target gates in a reversible circuit is linear as described above.  The
algorithms presented in this paper do not necessarily consider a specific target
technology. Hence, circuits obtained from our synthesis method are a reasonable
intermediate gate-level representation for large reversible
functions~\citep{ASTD:2014}.  It remains for future work to show how appropriate
technology mappings can be derived from single-target gates.  In the
experimental evaluation we are making use of straight-forward mapping
techniques.

\subsection{Related work}
In~\citep{SWH+:2012} another algorithm for synthesizing large reversible
functions without adding ancilla lines has been proposed that is based on
quantum multiple-valued decision diagrams~\citep{MT:2006}.  Compared to the
approach presented in the present paper, this algorithm uses a different method
to determine Toffoli gates.  Our experimental evaluations show that with our
technique smaller circuits with respect to the number of Toffoli gates and
quantum cost can be found.  The algorithm proposed in~\citep{WD:2009} also uses
BDDs to find reversible circuits.  However, the algorithm uses irreversible
functions as input and embeds them into reversible functions implicitly using a
hierarchical approach.  The approach produces an enormous number of additional
helper lines which are still far beyond from the theoretical upper bound.  In
contrast, the algorithms presented in this paper make use of a scalable exact
embedding algorithm~\citep{SWG+:14} and hence guarantee synthesis without adding
ancilla lines.

In~\citep{SZSS:10} a cycle-based approach to synthesize reversible functions is
proposed. In this approach a decomposition algorithm is first used to extract
building blocks from the input specification of the function. The input
specification is represented as a permutation and is decomposed into smaller
permutations, so called $k$-cycles, until building blocks can be used for
synthesis.  However, this method needs all input assignments to derive the
required $k$-cycles and as a consequence it has the same limitations as truth
table based approaches and cannot be applied to large functions. In our approach
the function is decomposed symbolically using the co-factor representation of
the function obtained from a BDD, which allows us to find cycles without
necessarily traversing the whole input space.  A similar approach has been
presented in~\citep{SSSZ:09} but shares the same limitations as it is also based
on an exponential function representation.

\subsection{Outline}
The remainder of the paper is organized as follows.  Preliminaries are given in
the next section and Sect.~\ref{sec:truth-table-based} reviews the truth table
based decomposition technique from \citep{VR:2008}.
Section~\ref{sec:general-idea} illustrates the general idea of the proposed
algorithm while Sect.~\ref{sec:char-repr-revers} shows special BDD operations
that are commonly used in the description of the algorithm in
Sect.~\ref{sec:algorithm}.  Optimization techniques targeting efficiency and the
number of gates are presented in Sects.~\ref{sec:incr-effic}
and~\ref{sec:ordering}, respectively.  The algorithm is extended for partial
functions in Sect.~\ref{sec:partial-functions}.  Experimental evaluations are
presented in Sect.~\ref{sec:exper-eval} before the paper is concluded in
Sect.~\ref{sec:conclusions}.

\section{Preliminaries}
In this section, we introduce notation for Boolean functions, binary decision
diagrams, reversible functions, and reversible circuits.
\subsection{Boolean Functions}
\label{sec:bool-revers-funct}
\begin{definition}[Boolean function]
Let~$\bbbb\eqdef\{0,1\}$ denote the \emph{Boolean values}.
Then we refer to
\begin{equation}
\bfunc_{n,m}\eqdef\{f\mid f\colon\bbbb^n\to\bbbb^m\}
\end{equation}
as the set of all \emph{Boolean multiple-output functions} with $n$~inputs and
$m$~outputs, where $m,n\geq1$.
\end{definition}

\sn We write~$\bfunc_n\eqdef\bfunc_{n,1}$ for each~$n\geq1$ and assume that each
$f\in\bfunc_n$ is represented by a propositional formula over the input
variables $x_1,\dots,x_n$.  Furthermore, we assume that each function~$f\in
\bfunc_{n,m}$ is represented as a tuple $f=(f_1,\dots,f_m)$
where~$f_i\in\bfunc_n$ for each~$i\in\{1,\dots,m\}$ and hence~$f(\vec
x)=(f_1(\vec x),\dots,f_m(\vec x))$ for each~$\vec x\in\bbbb^n$.  Output
variables of a function are denoted~$y_1,\dots,y_m$.

\begin{figure}
  \def\tabcolsep{1.15pt}
  \def\arraystretch{1}
  \centering
  \begin{tabular}{|CC|CC|} \hline
    x_1 & x_2 & y_1 & y_2 \\ \hline
      0 &   0 &   0 &   0 \\
      0 &   1 &   1 &   0 \\
      1 &   0 &   1 &   0 \\
      1 &   1 &   0 &   1 \\ \hline
  \end{tabular}
  \caption{Truth table for the function~$f(x_1,x_2)=(x_1\oplus x_2,x_1\land x_2)$}
  \label{fig:truth-table-half-adder}
\end{figure}

\begin{example}
  Figure~\ref{fig:truth-table-half-adder} shows the truth table of the Boolean
  multiple-output function
  \begin{equation}
    \label{eq:ha}
    f(x_1,x_2) = (x_1\oplus x_2,x_1\land x_2)
  \end{equation}
  that has two input and two output variables.  It represents the functionality
  of a half-adder.
\end{example}

Boolean matrices can be represented by a Boolean
function as described by the following definition.

\begin{definition}[Boolean matrix]
  \label{def:boolean-matrix}
  A Boolean-valued $2^n\times 2^m$ matrix~$A$ can be represented by a Boolean
  function~$f_A\in\bfunc_{m+n}$ with
  \begin{equation}
    \label{eq:boolean-matrix}
    f_A(c_1,\dots,c_m,r_1,\dots,r_n)\eqdef A_{r,c}
  \end{equation}
  where~$A_{r,c}$ denotes the element at row~$r=\sum_{i=1}^n2^{n-i}r_i$ and
  column~$c=\sum_{i=1}^m2^{m -1}c_i$.
\end{definition}

\begin{example}
  The Boolean matrix
  \[
    A=\begin{pmatrix} 1&0&0&0 \\[-5pt] 0&0&0&1 \\[-5pt] 0&1&1&0 \\[-5pt] 0&0&0&0 \end{pmatrix}
  \]
  is represented by the Boolean function
  \begin{equation}
    \label{eq:mat-func}
   f_A(c_1,c_2,r_1,r_2) = \bar c_1\bar c_2\bar r_1\bar r_2
                    \lor \bar c_1     c_2     r_1\bar r_2
                    \lor      c_1\bar c_2     r_1\bar r_2
                    \lor      c_1     c_2\bar r_1     r_2.
  \end{equation}
\end{example}

\begin{definition}[Co-factors]
  Given a Boolean function~$f\in\bfunc_n$ over the variables~$x_1,\dots,x_n$ and
  a variable~$x_i$ we define the \emph{positive
    co-factor}~$f_{x_i}\in\bfunc_{n-1}$ and the \emph{negative
    co-factor}~$f_{\bar x_i}\in\bfunc_{n-1}$ as
\begin{equation}
  \label{eq:pos-cofactor}
  f_{x_i}\eqdef f(x_1,\dots,x_{i-1},1,x_{i+1},\dots,x_n)
\end{equation}
and
\begin{equation}
  \label{eq:neg-cofactor}
  f_{\bar x_i}\eqdef f(x_1,\dots,x_{i-1},0,x_{i+1},\dots,x_n),
\end{equation}
respectively.
\end{definition}

\begin{definition}[Smoothing operator]
Given a Boolean function~$f\in\bfunc_n$ and an input~$x_i$ of~$f$, the
\emph{smoothing operator}~$\exists x_i$~\citep{TSL+:1990} is defined as the
disjunction of both co-factors, i.e.
\begin{equation}
  \label{eq:smoothing}
  \exists x_i\,f\eqdef f_{\bar x_i}\lor f_{x_i}
\end{equation}
We denote~$\exists\vec x\,f\eqdef\exists x_1\cdots\exists x_n\, f$.
\end{definition}

\sn That is, the smoothing operator returns a function that does not depend on
the variable~$x_i$ anymore.  Informally one can describe the smoothing operator
by replacing all occurrences of~$x_i$ and~$\bar x_i$ with don't cares.  The
smoothing operator can e.g.~be used for matrix multiplication as illustrated by
the following lemma.

\begin{lemma}[\citeauthor{TSL+:1990}, \citeyear{TSL+:1990}]
  \label{lem:multiply}
  Let~$A$ be a~$2^k\times 2^n$ Boolean matrix and~$B$ be a $2^m\times 2^k$
  Boolean matrix that are represented by Boolean
  functions~$f_A(x_1,\dots,x_n,y_1,\dots,y_k)$ and~$f_B(y_1,\dots,y_k,
  z_1,\dots,z_m)$ respectively.  Let
  \[ f_C(x_1,\dots,x_n,z_1,\dots,z_m)=\exists y_1\dots\exists y_n\,(f_A\land
  f_B). \] Then~$f_C$ is the Boolean function representation for the
  $2^m\times2^n$ Boolean matrix~$C=B\cdot A$ where~`$\cdot$' is the matrix
  multiplication in the Galois field~$\bbbf_2$.  \hfill$\Box$
\end{lemma}

\begin{definition}[ON-set and OFF-set]
  Given a Boolean function~$f\in\bfunc_n$ the sets
  \begin{equation}
    \onset(f)\eqdef\{\vec x\in\bbbb^n\mid f(\vec x)=1\}
    \quad\text{and}\quad
    \offset(f)\eqdef\{\vec x\in\bbbb^n\mid f(\vec x)=0\}
  \end{equation}
  are called \emph{ON-set} and \emph{OFF-set} of~$f$.  It can easily be seen
  that~$\onset(f)\cap\offset(f)=\emptyset$ and
  $\onset(f)\cup\offset(f)=\bbbb^n$.
\end{definition}

\begin{definition}[Characteristic function]
Given a function~$f=(f_1,\dots,f_m)\in\bfunc_{n,m}$ its \emph{characteristic
function}~$\chi_f\in\bfunc_{n+m}$ is defined as
\begin{equation}
  \label{eq:characteristic}
  \chi_f(\vec x,\vec y)\quad\eqdef\quad
\begin{cases} 1 & f(\vec x)=\vec y\\
0 & \text{otherwise}
\end{cases}
\end{equation}
for each~$\vec x\in\bbbb^n$ and each~$\vec y\in\bbbb^m$.
\end{definition}

\sn
The characteristic function allows to represent any multiple-output function as
a single-output function.
It can be computed from a multiple-output function by adding
to the variables $\{x_1,\ldots,x_n\}$ the
additional output variables~$\{y_1,\dots,y_m\}$:
\begin{equation}
  \label{eq:characteristic_computation}
\bigwedge_{i=1}^m(y_i\leftrightarrow f_i(x_1,\dots,x_n))
\end{equation}
In the remainder of this paper, we denote the characteristic function~$\chi_f$
of a function~$f$ by a capital letter, i.e.~$F$.

\begin{example}
  The function~$f_A(x_1,x_2,y_1,y_2)$ in~\eqref{eq:mat-func} is the
  characteristic function of \linebreak $f(x_1,x_2)=(y_1,y_2)$ in~\eqref{eq:ha}.
\end{example}

\subsection{Binary Decision Diagrams}
\label{sec:binary-decis-diagr}
Binary decision diagrams (BDD) are an established data structure for representing
Boo\-lean functions.  While the general concepts are briefly outlined in this
section, the reader is referred to the literature for a comprehensive
over\-view~\citep{Bryant:1986,Knuth:2011}.

Let~$\vec x=x_1,\dots,x_n$ be the variables of a Boolean function
$f\in\bfunc_n$.  A BDD representing the function $f$ is a directed acyclic graph
with non-terminal vertices~$N$ and terminal vertices~$T\subseteq\{\bddFalse,
\bddTrue\}$ where $N\cap T=\emptyset$ and~$T\neq\emptyset$.  Each non-terminal
vertex~$v\in N$ is labeled by a variable from $\vec x$ and has exactly two
children, $\low v $ and~$\high v $.  The directed edges to these children are
called \emph{low-edge} and \emph{high-edge} and are drawn dashed and solid,
respectively.  A non-terminal vertex $v$ labeled $x_i$ represents a function
denoted~$\sigma(v)$ given by the \emph{Shannon
  decomposition}~\citep{Shannon:1938}
\begin{equation}
\label{eq:shannon}
\sigma(v)=\bar x_i\sigma(\low v)+x_i\sigma(\high v)
\end{equation}
\noindent where $\sigma(\low v)$ and $\sigma(\high v)$ are the functions
represented by the children of $v$ with $\sigma(\bddFalse)=0$ and
$\sigma(\bddTrue)=1$.  The BDD has a single start vertex~$s$ with $\sigma(s)=f$.

A BDD is \emph{ordered} if the variables of the vertices on every path from the
start vertex to a terminal vertex adhere to a specific ordering.
Not all of the variables need to appear on a particular path, but a variable can
appear at most once on any path.
A BDD is \emph{reduced} if there are no two non-terminal vertices representing
the same function, hence the representation of common subfunctions is shared.
In the following only reduced, ordered BDDs are considered and for briefness
referred to as~BDDs.

Multiple-output functions can be represented by a single BDD that has more than
one start vertex.  Common subfunctions that can be shared among the functions
decrease the overall size of the BDD.  In fact, many practical Boolean functions
can efficiently be represented using BDDs, and efficient manipulations and
evaluations are possible.

\begin{figure}[t]
  \centering
  \subfloat[Shannon decomposition]{%
    \begin{tikzpicture}
      \node [draw,ellipse,inner sep=1.5pt] (x) {$x_i$};
      \draw[dashed] (x) to[bend right=10] +(225:.8cm) node[below] {$\low v$};
      \draw         (x) to[bend left=10]  +(315:.8cm) node[below] {$\high v$};
      \draw         (x) to                +(90:.4cm)  node[above] {$\sigma(v)=\bar x_i\sigma(\low v)+x_i\sigma(\high v)$};
    \end{tikzpicture}
  }
  \subfloat[BDD for function in Fig.~\ref{fig:truth-table-half-adder}]{%
  \begin{tikzpicture}[scale=.5]
    \begin{scope}[every node/.style={draw,ellipse,inner sep=1.5pt}]
      \node (1) at (41bp,162bp) [] {$x_1$};
      \node (3) at (27bp,90bp) [] {$x_2$};
      \node (2) at (99bp,90bp) [] {$x_2$};
      \node (4) at (140bp,162bp) [] {$x_1$};
    \end{scope}
    \begin{scope}[every node/.style={draw,inner sep=1.5pt}]
      \node (c1) at (27bp,18bp) [] {$\top$};
      \node (c0) at (99bp,18bp) [] {$\bot$};
    \end{scope}
    \draw [dashed] (4) ..controls (142.9bp,125.44bp) and (143.09bp,95.894bp)  .. (135bp,72bp) .. controls (130.58bp,58.935bp) and (121.99bp,46.092bp)  .. (c0);
    \draw [] (2) ..controls (72.924bp,63.649bp) and (56.827bp,47.999bp)  .. (c1);
    \draw [] (3) ..controls (53.076bp,63.649bp) and (69.173bp,47.999bp)  .. (c0);
    \draw [dashed] (2) ..controls (99bp,60.846bp) and (99bp,46.917bp)  .. (c0);
    \draw [dashed] (3) ..controls (27bp,60.846bp) and (27bp,46.917bp)  .. (c1);
    \draw [dashed] (1) ..controls (63.314bp,134.07bp) and (76.717bp,117.89bp)  .. (2);
    \draw [] (1) ..controls (35.442bp,133.21bp) and (32.637bp,119.18bp)  .. (3);
    \draw [] (4) ..controls (123.87bp,133.46bp) and (115.07bp,118.44bp)  .. (2);
    \draw (1.north) -- ++(up:10pt) node[above] {$y_1$};
    \draw (4.north) -- ++(up:10pt) node[above] {$y_2$};
    \coordinate (c) at (current bounding box.center);
    \draw[opacity=0] ([xshift=-130pt] c) -- ++(right:260pt);
  \end{tikzpicture}}
  \caption{Binary decision diagrams}
  \label{fig:bdds}
\end{figure}
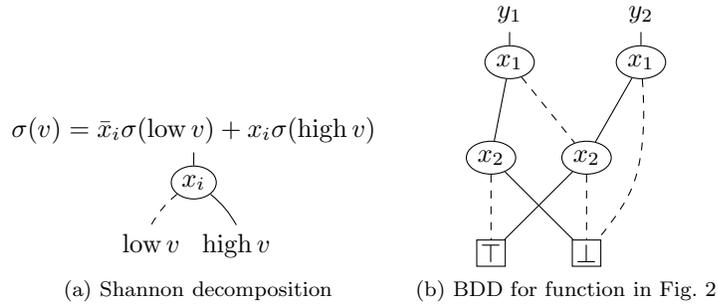

\begin{example}
  Figure~\ref{fig:bdds}(a) illustrates the Shannon decomposition
  from~\eqref{eq:shannon}.  A binary decision diagram for the function in
  Fig.~\ref{fig:truth-table-half-adder} is given in Fig.~\ref{fig:bdds}(b).
\end{example}

\subsection{Reversible Boolean Functions}
\begin{definition}[Reversible function]
A function~$f\in\bfunc_{n,m}$ is called \emph{reversible} if~$f$ is bijective,
otherwise it is called \emph{irreversible}.
Clearly, if~$f$ is reversible, then~$n=m$.
\end{definition}

\sn
A reversible function~$f\in\bfunc_{n,n}$ can also be represented by a
\emph{permutation} of~$\{0,1,\dots,2^n-1\}$, i.e.
\begin{equation}
  \label{eq:permutation}
  \pi_f\eqdef\left(\nat(f(0,\dots,0,0)),\dots,\nat(f(1,\dots,1,1))\right),
\end{equation}
where~$\nat:\bbbb^n\to\{0,1,\dots,2^n-1\}$ maps a bit-vector to its natural
number representation.
Further, $f$~can be represented as a~$2\times2$ Boolean \emph{permutation
matrix} $\Pi_f$ where
\begin{equation}
  \label{eq:permutation-matrix}
  \left(\Pi_f\right)_{r,c}\eqdef\left(\pi_f(c)\equiv r\right)
\end{equation}

\sn
The following lemma relates the characteristic function to the permutation
matrix of a reversible function.

\begin{lemma}
  \label{lem:char-perm-matrix}
  Let~$f\in\bfunc_{n,n}$ be a reversible function, $\Pi_f$ the permutation matrix
  of~$f$, and~$F$ the characteristic function of~$f$.
  Then~$\Pi_f$ is the matrix representation of~$F$ according to
  Definition~\ref{def:boolean-matrix}.
  We have~$|\onset(F)|=2^n$.
\end{lemma}

\subsection{Reversible Circuits}
Reversible functions can be realized by reversible circuits that consist of at
least~$n$ lines and are constructed as cascades of reversible gates that belong
to a certain gate library.  Single-target gates have already been defined in the
introduction.  The most common gate library consists of Toffoli gates.

\begin{definition}[Toffoli gate]
  \emph{Mixed-polarity multiple-control Toffoli~(MPMCT) gates} are a subset of
  the single-target gates in which the control function~$c$ can be represented
  with one product term or~$c=1$.  We refer to MPMCT gates as Toffoli gates in
  the following.  The literals in the control function are also referred to as
  \emph{controls} or \emph{control lines}.
\end{definition}

\noindent
In~\citep{SPMH:2003}, it has been shown that any reversible
function~$f\in\bfunc_{n,n}$ can be realized by a reversible circuit with~$n$
lines when using Toffoli gates.  That is, it is not necessary to add any
temporary lines (ancilla) to realize the circuit.  Note that each single-target
gate can be expressed in terms of a cascade of Toffoli gates, which can be
obtained from an ESOP expression~\citep{Sas:93c}, respectively.  For drawing
circuits, we follow the established conventions of using the symbol $\oplus$ to
denote the target line, solid black circles to indicate positive control lines
and white circles to indicate negative control lines.

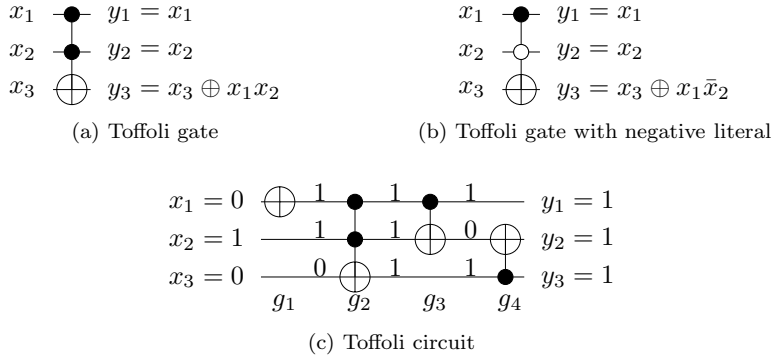
\begin{figure}[t]
\centering
  \subfloat[Toffoli gate]{%
    \begin{tikzpicture}
      \draw[line width=0.300000] (0.500000,1.250000) -- (1.000000,1.250000);
      \draw (0.400000,1.250000) node [left] {$x_1$};
      \draw (1.100000,1.250000) node [right] {$y_1=x_1$};
      \draw[line width=0.300000] (0.500000,0.750000) -- (1.000000,0.750000);
      \draw (0.400000,0.750000) node [left] {$x_2$};
      \draw (1.100000,0.750000) node [right] {$y_2=x_2$};
      \draw[line width=0.300000] (0.500000,0.250000) -- (1.000000,0.250000);
      \draw (0.400000,0.250000) node [left] {$x_3$};
      \draw (1.100000,0.250000) node [right] {$y_3=x_3\oplus x_1x_2$};
      \draw[line width=0.300000] (0.750000,0.250000) -- (0.750000,1.250000);
      \draw[fill] (0.750000,1.250000) circle (0.100000);
      \draw[fill] (0.750000,0.750000) circle (0.100000);
      \draw[line width=0.3] (0.75,0.25) circle (0.2) (0.75,0.05) -- (0.75,0.45);
    \end{tikzpicture}}\hfil
   \subfloat[Toffoli gate with negative literal]{%
    \begin{tikzpicture}
      \draw[line width=0.300000] (0.500000,1.250000) -- (1.000000,1.250000);
      \draw (0.400000,1.250000) node [left] {$x_1$};
      \draw (1.100000,1.250000) node [right] {$y_1=x_1$};
      \draw[line width=0.300000] (0.500000,0.750000) -- (1.000000,0.750000);
      \draw (0.400000,0.750000) node [left] {$x_2$};
      \draw (1.100000,0.750000) node [right] {$y_2=x_2$};
      \draw[line width=0.300000] (0.500000,0.250000) -- (1.000000,0.250000);
      \draw (0.400000,0.250000) node [left] {$x_3$};
      \draw (1.100000,0.250000) node [right] {$y_3=x_3\oplus x_1\bar x_2$};
      \draw[line width=0.300000] (0.750000,0.250000) -- (0.750000,1.250000);
      \draw[fill] (0.750000,1.250000) circle (0.100000);
      \draw[fill=white] (0.750000,0.750000) circle (0.100000);
      \draw[line width=0.3] (0.75,0.25) circle (0.2) (0.75,0.05) -- (0.75,0.45);
      \coordinate (c) at (current bounding box.center);
      \draw[opacity=0,line width=0pt] ([xshift=-2.5cm] c) -- ++(right:5cm);
    \end{tikzpicture}}

  \subfloat[Toffoli circuit]{%
\begin{tikzpicture}
\draw[line width=0.300000] (0.500000,1.250000) -- (4.00000,1.250000);
\draw (0.400000,1.250000) node [left] {$x_1 = 0$};
\draw (4.100000,1.250000) node [right] {$y_1 =1$};
\draw[line width=0.300000] (0.500000,0.750000) -- (4.00000,0.750000);
\draw (0.400000,0.750000) node [left] {$x_2 = 1$};
\draw (4.100000,0.750000) node [right] {$y_2 = 1$};
\draw[line width=0.300000] (0.500000,0.250000) -- (4.00000,0.250000);
\draw (0.400000,0.250000) node [left] {$x_3 = 0$};
\draw (4.100000,0.250000) node [right] {$y_3 = 1$};
\draw[line width=0.300000] (0.750000,1.250000) -- (0.750000,1.250000);
\draw[line width=0.3] (0.75,1.25) circle (0.2) (0.75,1.05) -- (0.75,1.45);
\draw (1.100000,-0.100000) node [left] {$g_1$};
\draw (1.500000,1.380000) node [left] {$1$};
\draw (1.500000,0.880000) node [left] {$1$};
\draw (1.500000,0.380000) node [left] {$0$};
\draw[line width=0.300000] (1.750000,0.250000) -- (1.750000,1.250000);
\draw[fill] (1.750000,1.250000) circle (0.100000);
\draw[fill] (1.750000,0.750000) circle (0.100000);
\draw[line width=0.3] (1.75,0.25) circle (0.2) (1.75,0.05) -- (1.75,0.45);
\draw (2.100000,-0.100000) node [left] {$g_2$};
\draw (2.500000,1.380000) node [left] {$1$};
\draw (2.500000,0.880000) node [left] {$1$};
\draw (2.500000,0.380000) node [left] {$1$};
\draw[line width=0.300000] (2.750000,0.750000) -- (2.750000,1.250000);
\draw[fill] (2.750000,1.250000) circle (0.100000);
\draw[line width=0.3] (2.75,0.75) circle (0.2) (2.75,0.55) -- (2.75,0.95);
\draw (3.100000,-0.100000) node [left] {$g_3$};
\draw (3.500000,1.380000) node [left] {$1$};
\draw (3.500000,0.880000) node [left] {$0$};
\draw (3.500000,0.380000) node [left] {$1$};
\draw[line width=0.300000] (3.750000,0.250000) -- (3.750000,0.750000);
\draw[fill] (3.750000,0.250000) circle (0.100000);
\draw[line width=0.3] (3.75,0.75) circle (0.2) (3.75,0.55) -- (3.75,0.95);
\draw (4.100000,-0.100000) node [left] {$g_4$};
\end{tikzpicture}}
\caption{Reversible circuitry}
\label{fig:reversible-circuits}
\end{figure}

\begin{example}
  Figure~\ref{fig:reversible-circuits}(a) shows a Toffoli gate with two positive
  literals, while Fig.~\ref{fig:reversible-circuits}(b) shows a Toffoli gate
  with mixed polarities.  Figure~\ref{fig:reversible-circuits}(c) shows four
  Toffoli gates in a cascade forming a reversible circuit.  The annotated values
  demonstrate the computation of the gate for a given input assignment.
\end{example}

\section{Truth table based Decomposition}
\label{sec:truth-table-based}

\begin{figure}[t]
  \def\tabcolsep{1.15pt}
  \def\arraystretch{1}
  \newcommand{\pr}[2]{\tikz[remember picture,baseline=(n#2.base)]%
    \node[inner sep=1pt,outer sep=0pt,draw] (n#2) {#1};}
  \newcommand{\pn}[2]{\tikz[remember picture,baseline=(n#2.base)]%
    \node[inner sep=1pt,outer sep=0pt] (n#2) {#1};}
  \subfloat[Initial truth table]{%
    \begin{tabular}{|CCC|CCC|} \hline
      x_1&x_2&x_3&y_1&y_2&y_3 \\ \hline
      0 & 0 & 0 & 0 & 0 & 0 \\
      0 & 0 & 1 & 0 & 0 & 1 \\
      0 & 1 & 0 & 0 & 1 & 0 \\
      0 & 1 & 1 & 1 & 0 & 0 \\
      1 & 0 & 0 & 1 & 0 & 1 \\
      1 & 0 & 1 & 0 & 1 & 1 \\
      1 & 1 & 0 & 1 & 1 & 1 \\
      1 & 1 & 1 & 1 & 1 & 0 \\ \hline
    \end{tabular}
  }\hfill
  \subfloat[Copy second and third variable]{%
    \begin{tabular}{|CCC|CCC|CCC|CCC|} \hline
      x_1&x_2&x_3&x_1'&x_2&x_3&y_1'&y_2&y_3&y_1&y_2&y_3 \\ \hline
      0 & 0 & 0 &   & 0 & 0 &   & 0 & 0 & 0 & 0 & 0 \\
      0 & 0 & 1 &   & 0 & 1 &   & 0 & 1 & 0 & 0 & 1 \\
      0 & 1 & 0 &   & 1 & 0 &   & 1 & 0 & 0 & 1 & 0 \\
      0 & 1 & 1 &   & 1 & 1 &   & 0 & 0 & 1 & 0 & 0 \\
      1 & 0 & 0 &   & 0 & 0 &   & 0 & 1 & 1 & 0 & 1 \\
      1 & 0 & 1 &   & 0 & 1 &   & 1 & 1 & 0 & 1 & 1 \\
      1 & 1 & 0 &   & 1 & 0 &   & 1 & 1 & 1 & 1 & 1 \\
      1 & 1 & 1 &   & 1 & 1 &   & 1 & 0 & 1 & 1 & 0 \\ \hline
    \end{tabular}
  }\hfill
  \subfloat[Fill in first four numbers]{%
    \begin{tabular}{|CCC|CCC|CCC|CCC|} \hline
      x_1&x_2&x_3&x_1'&x_2&x_3&y_1'&y_2&y_3&y_1&y_2&y_3 \\ \hline
      0 & 0 & 0 & \pr1a & 0 & 0 &       & 0 & 0 & 0 & 0 & 0 \\
      0 & 0 & 1 &       & 0 & 1 & \pn1d & 0 & 1 & 0 & 0 & 1 \\
      0 & 1 & 0 &       & 1 & 0 &       & 1 & 0 & 0 & 1 & 0 \\
      0 & 1 & 1 &       & 1 & 1 &       & 0 & 0 & 1 & 0 & 0 \\
      1 & 0 & 0 & \pn0b & 0 & 0 & \pn0c & 0 & 1 & 1 & 0 & 1 \\
      1 & 0 & 1 &       & 0 & 1 &       & 1 & 1 & 0 & 1 & 1 \\
      1 & 1 & 0 &       & 1 & 0 &       & 1 & 1 & 1 & 1 & 1 \\
      1 & 1 & 1 &       & 1 & 1 &       & 1 & 0 & 1 & 1 & 0 \\ \hline
    \end{tabular}
    \begin{tikzpicture}[remember picture,overlay]
      \draw[-latex] (na.south) -- (nb.north);
      \draw[-latex] (nb.east) -- ++(right:3pt) to[out=0,in=180] ++(5pt,-5pt)
      -- ([xshift=-8pt,yshift=-5pt] nc.west) to[out=0,in=180] ++(5pt,5pt)
      -- (nc.west);
      \draw[-latex] (nc.north) -- (nd.south);
    \end{tikzpicture}
  }

  \subfloat[Fill in remaining numbers]{%
    \begin{tabular}{|CCC|CCC|CCC|CCC|} \hline
      x_1&x_2&x_3&x_1'&x_2&x_3&y_1'&y_2&y_3&y_1&y_2&y_3 \\ \hline
      0 & 0 & 0 & 1 & 0 & 0 & 1 & 0 & 0 & 0 & 0 & 0 \\
      0 & 0 & 1 & 1 & 0 & 1 & 1 & 0 & 1 & 0 & 0 & 1 \\
      0 & 1 & 0 & 0 & 1 & 0 & 0 & 1 & 0 & 0 & 1 & 0 \\
      0 & 1 & 1 & 0 & 1 & 1 & 0 & 0 & 0 & 1 & 0 & 0 \\
      1 & 0 & 0 & 0 & 0 & 0 & 0 & 0 & 1 & 1 & 0 & 1 \\
      1 & 0 & 1 & 0 & 0 & 1 & 0 & 1 & 1 & 0 & 1 & 1 \\
      1 & 1 & 0 & 1 & 1 & 0 & 1 & 1 & 1 & 1 & 1 & 1 \\
      1 & 1 & 1 & 1 & 1 & 1 & 1 & 1 & 0 & 1 & 1 & 0 \\ \hline
    \end{tabular}
  }\hfill
  \subfloat[Identifying cubes]{%
    \begin{tabular}{|CCC|CCC|CCC|CCC|} \hline
      x_1&x_2&x_3&x_1'&x_2&x_3&y_1'&y_2&y_3&y_1&y_2&y_3 \\ \hline
      0 & \pn0a & 0     & 1 & 0 & 0 & 1 & 0 & 0 & 0 & \pn0e & 0     \\
      0 & 0     & \pn1b & 1 & 0 & 1 & 1 & 0 & 1 & 0 &     0 & \pn1f \\
      0 & 1     & 0     & 0 & 1 & 0 & 0 & 1 & 0 & 0 &     1 & 0     \\
      0 & 1     & 1     & 0 & 1 & 1 & 0 & 0 & 0 & 1 & \pn0g & 0     \\
      1 & \pn0c & 0     & 0 & 0 & 0 & 0 & 0 & 1 & 1 &     0 & \pn1h \\
      1 & 0     & \pn1d & 0 & 0 & 1 & 0 & 1 & 1 & 0 &     1 & 1     \\
      1 & 1     & 0     & 1 & 1 & 0 & 1 & 1 & 1 & 1 &     1 & 1     \\
      1 & 1     & 1     & 1 & 1 & 1 & 1 & 1 & 0 & 1 &     1 & 0     \\ \hline
    \end{tabular}
    \begin{tikzpicture}[remember picture,overlay]
      \draw[rounded corners=2pt] (na.north west) rectangle (nb.south east);
      \draw[rounded corners=2pt] (nc.north west) rectangle (nd.south east);
      \draw[rounded corners=2pt] (ne.north west) rectangle (nf.south east);
      \draw[rounded corners=2pt] (ng.north west) rectangle (nh.south east);
    \end{tikzpicture}
  }

  \subfloat[Equalizing $x_2$ and $y_2$]{%
    \begin{tabular}{|CCC|CCC|CCC|CCC|CCC|CCC|} \hline
      x_1&x_2&x_3&x_1'&x_2&x_3&x_1'&x_2'&x_3&
      y_1'&y_2'&y_3&y_1'&y_2&y_3&y_1&y_2&y_3 \\ \hline
    0 & 0 & 0 & 1 & 0 & 0 & 1 & \pr0a & 0 & 1 & 0 & 0 & 1 & 0 & 0 & 0 & 0 & 0 \\
    0 & 0 & 1 & 1 & 0 & 1 & 1 &    0  & 1 & 1 & 0 & 1 & 1 & 0 & 1 & 0 & 0 & 1 \\
    0 & 1 & 0 & 0 & 1 & 0 & 0 & \pr0b & 0 & 0 & 0 & 0 & 0 & 1 & 0 & 0 & 1 & 0 \\
    0 & 1 & 1 & 0 & 1 & 1 & 0 &    1  & 1 & 0 & 1 & 0 & 0 & 0 & 0 & 1 & 0 & 0 \\
    1 & 0 & 0 & 0 & 0 & 0 & 0 &    1  & 0 & 0 & 1 & 1 & 0 & 0 & 1 & 1 & 0 & 1 \\
    1 & 0 & 1 & 0 & 0 & 1 & 0 &    0  & 1 & 0 & 0 & 1 & 0 & 1 & 1 & 0 & 1 & 1 \\
    1 & 1 & 0 & 1 & 1 & 0 & 1 &    1  & 0 & 1 & 1 & 1 & 1 & 1 & 1 & 1 & 1 & 1 \\
    1 & \pn1c & 1 & 1 & \pn1d & 1 & 1 &    \pn1e  & 1 & 1 & \pn1f & 0 & 1 & \pn1g & 0 & 1 & \pn1h & 0 \\
      \hline
      \noalign{\vspace{1cm}}
    \end{tabular}
    \begin{tikzpicture}[remember picture,overlay,font=\footnotesize]
      \draw[latex-latex] ([xshift=5pt,yshift=-3pt] nc.south)
                         to[out=270,in=270] node[below] {$\T[\bar x_2,1]$}
                         ([xshift=-5pt,yshift=-3pt] nd.south);
      \draw[latex-latex] ([xshift=5pt,yshift=-3pt] nd.south)
                         to[out=270,in=270] node[below] {$\T[\bar x_1'\bar x_3,2]$}
                         ([xshift=-5pt,yshift=-3pt] ne.south);
      \draw[latex-latex] ([xshift=5pt,yshift=-3pt] ne.south)
                         to[out=270,in=270] node[below] {$\T[x_2',3]$}
                         ([xshift=-5pt,yshift=-3pt] nf.south);
      \draw[latex-latex] ([xshift=5pt,yshift=-3pt] nf.south)
                         to[out=270,in=270] node[below] {$\T[\bar y_1',2]$}
                         ([xshift=-5pt,yshift=-3pt] ng.south);
      \draw[latex-latex] ([xshift=5pt,yshift=-3pt] ng.south)
                         to[out=270,in=270] node[below] {$\T[\bar y_2,1]$}
                         ([xshift=-5pt,yshift=-3pt] nh.south);
    \end{tikzpicture}
  }\hfill
  \subfloat[Resulting circuit]{%
    \begin{tikzpicture}
      \draw[line width=0.300000] (0.500000,1.250000) -- (4.600000,1.250000);
      \draw (0.400000,1.250000) node [left] {$x_1$};
      \draw (4.700000,1.250000) node [right] {$y_1$};
      \draw[line width=0.300000] (0.500000,0.750000) -- (4.600000,0.750000);
      \draw (0.400000,0.750000) node [left] {$x_2$};
      \draw (4.700000,0.750000) node [right] {$y_2$};
      \draw[line width=0.300000] (0.500000,0.250000) -- (4.600000,0.250000);
      \draw (0.400000,0.250000) node [left] {$x_3$};
      \draw (4.700000,0.250000) node [right] {$y_3$};
      \draw[line width=0.300000] (0.750000,0.750000) -- (0.750000,1.250000);
      \draw[fill=white] (0.750000,0.750000) circle (0.100000);
      \draw[line width=0.3] (0.75,1.25) circle (0.2) (0.75,1.05) -- (0.75,1.45);
      \draw[line width=0.300000] (1.650000,0.250000) -- (1.650000,1.250000);
      \draw[fill=white] (1.650000,1.250000) circle (0.100000);
      \draw[fill=white] (1.650000,0.250000) circle (0.100000);
      \draw[line width=0.3] (1.65,0.75) circle (0.2) (1.65,0.55) -- (1.65,0.95);
      \draw[line width=0.300000] (2.550000,0.250000) -- (2.550000,0.750000);
      \draw[fill] (2.550000,0.750000) circle (0.100000);
      \draw[line width=0.3] (2.55,0.25) circle (0.2) (2.55,0.05) -- (2.55,0.45);
      \draw[line width=0.300000] (3.450000,0.750000) -- (3.450000,1.250000);
      \draw[fill=white] (3.450000,1.250000) circle (0.100000);
      \draw[line width=0.3] (3.45,0.75) circle (0.2) (3.45,0.55) -- (3.45,0.95);
      \draw[line width=0.300000] (4.350000,0.750000) -- (4.350000,1.250000);
      \draw[fill=white] (4.350000,0.750000) circle (0.100000);
      \draw[line width=0.3] (4.35,1.25) circle (0.2) (4.35,1.05) -- (4.35,1.45);
      \node[font=\footnotesize] at (1.2,1.41) {$x_1'$};
      \node[font=\footnotesize] at (2.1,0.92) {$x_2'$};
      \node[font=\footnotesize] at (3,0.92) {$y_2'$};
      \node[font=\footnotesize] at (3.9,1.41) {$y_1'$};
    \end{tikzpicture}
  }
  \caption{Truth table based decomposition}
  \label{fig:devos-example}
\end{figure}
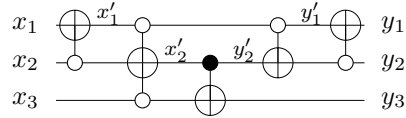
\label{sec:devos}
In~\citep{VR:2008} an algorithm was presented that determines control
functions~$l$ and~$r$ as in~\eqref{eq:decomposition} to decompose a reversible
function~$f:\bbbb^n\to\bbbb^n$ that is represented as a truth table.  When
applying the algorithm for each circuit line one can obtain a reversible circuit
composed of single-target gates that realizes~$f$.  The underlying theory for
the algorithm is based on Young subgroups, however, for the scope of this paper
it is sufficient to explain the algorithm by means of an example.  Consider the
truth table in Fig.~\ref{fig:devos-example}(a) that is described by the
variables~$x_1,x_2,x_3$ and $y_1,y_2,y_3$.

First, the aim is to determine two control functions~$l$ and~$r$ for
single-target gates that act on the first circuit line which maps~$x_1$
to~$y_1$.  All values of the variables~$x_2,x_3$ and~$y_2,y_3$ are copied
into a new truth table that will eventually represent~$f'$ and is illustrated by
the two inner blocks in Fig.~\ref{fig:devos-example}(b).  Afterwards the values
of~$x_1'$ and~$y_1'$ are assigned new values such that the following two
constraints are met:
\begin{enumerate}
\item The functions defined by the truth table described by the first two blocks
  and by the last two blocks must be reversible, i.e.~each pattern must occur
  exactly once.
\item The function~$f'$, defined by the inner two blocks, must ensure that
  $x_1'=y_1'$.
\end{enumerate}
The reversibility of~$f'$ follows from both constraints and since both~$x_2,x_3$
and~$y_2,y_3$ were copied.  It also follows that the functions described by the
first and last two blocks can be described by single-target gates.

The values for~$x_1'$ and~$y_1'$ can be filled using the following procedure.
We start by inserting a value at the first empty cell of~$x_1'$, in our example
we choose the value~$1$ indicated by a rectangle in
Fig.~\ref{fig:devos-example}(c).  To satisfy the first constraint, next a~$0$ is
inserted where the pattern repeats for~$x_2$ and~$x_3$.  In order to satisfy the
second constraint, the~$0$ has to be assigned to~$y_1'$ in the same row.
Consequently, a~$1$ has to be filled for~$y_1'$ where the pattern repeats.  This
procedure is repeated until a cell is met that has already been filled.  If
there are still empty cells, the overall procedure is repeated from the
beginning.  All numbers have been filled in Fig.~\ref{fig:devos-example}(d).

\begin{remark}
\label{rem:freedom}
The first entry position in the procedure above can freely be chosen.  It can
easily be seen, that the corresponding cycle will not change, it is just
\emph{entered} at a different position.  Also the value of the first entry can
be changed without violating the constraints, if the other values are adjusted
accordingly.  Hence, one obtains up to~$2^k$ different functions for~$l$ and~$r$
where~$k$ is the number of cycles.
\end{remark}

\sn
After all values for~$x_1'$ and~$y_1'$ have been assigned, the control functions
can be determined as illustrated by means of Fig.~\ref{fig:devos-example}(e).
The first control function~$l$ can be determined from the first two blocks by
inspecting the assignments for~$x_2$ and~$x_3$ where~$x_1$ and~$x_1'$ are
different.  These are~$00$ and~$01$, hence~$l=\bar x_2\bar x_3\lor \bar
x_2x_3=\bar x_2$.  This can be done analogously for~$r$ and one obtains~$r=\bar
y_2$.

This process is called~\emph{variable equalization} in the following.  The two
blocks in the middle of Fig.~\ref{fig:devos-example}(e) are a truth table
representation of~$f'$ which can now be used to obtain two further single-target
gates by equalizing the variables~$x_2$ and~$y_2$.  This process is illustrated
in Fig.~\ref{fig:devos-example}(f) where the starting points for filling numbers
are again indicated by rectangles.  The control functions for the single-target
gates to equalize~$x_2$ and~$y_2$ are~$\bar x_1'\bar x_3$ and~$\bar y_1'$,
respectively.  Note that the last single-target gate to equalize~$x_3$
and~$y_3$ can now directly be read from the truth table described by the middle
two blocks in Fig.~\ref{fig:devos-example}(f), since the values in the first and
second variables are already equal.  The control function of the last
single-target gate is~$\bar x_2'$.

When putting all single-target gates together, one obtains a circuit that is
depicted in Fig.~\ref{fig:devos-example}(g).  Note that in this particular case,
all single-target gates are also Toffoli gates.

\section{General Idea}
\label{sec:general-idea}
This section illustrates how the decomposition procedure can be described
symbolically with Boolean function operations that can efficiently be
implemented using BDDs.  The problem with the current synthesis approach is the
underlying truth table representation which grows exponentially with respect to
the number of variables.  As a result, the processing time of the algorithm
depends essentially only on the size of the function and therefore impedes an
efficient processing of functions with many variables.

Aiming for reducing the complexity of the algorithm, we propose to implement the
algorithm based on BDDs, which provide a compact representation for many
reversible functions of practical interest.
Using BDDs allows for reducing the required memory as a first consequence.
Moreover, our proposed algorithm for decomposing a function according
to~\eqref{eq:decomposition} is constructed in a manner such that it does not
necessarily traverse the whole truth table; therefore reducing the run-time.

\begin{figure}[t]
\centering
\begin{tabular}{ccccc}
\hline
\multicolumn{1}{c}{} & $F_{n_1}$ & $F_{p'_1}$ & $F_{n'_1}$ & $F_{p_1}$ \\
\multicolumn{1}{c}{$x_1 \mapsto y_1$ } & $0 \mapsto 0$ & $1 \mapsto 0$  & $0 \mapsto 1$ & $1 \mapsto 1$  \\ \hline
\parbox[t]{2mm}{\multirow{4}{*}{\rotatebox[origin=c]{90}{Inputs}}} & $x_2x_3$  & $x_2 x_3$  & $x_2 x_3$ & $x_2 x_3$ \\ \cline{ 2- 5}\noalign{\vspace{2pt}}
\multicolumn{ 1}{c}{} & 00 & \tikz[baseline=(n.base)] \node[draw,inner sep=2pt] (n) {01}; & 11 & 00 \\
\multicolumn{ 1}{c}{} & 01 &    &    & \tikz[baseline=(n.base)] \node[draw,inner sep=2pt,rounded corners=2pt] (n) {10}; \\
\multicolumn{ 1}{c}{} & 10 &    &    & 11 \\ \hline
\parbox[t]{2mm}{\multirow{4}{*}{\rotatebox[origin=c]{90}{Outputs}}} & $y_2 y_3$  & $y_2 y_3$  & $y_2 y_3$  & $y_2 y_3$  \\ \cline{ 2- 5}\noalign{\vspace{2pt}}
\multicolumn{ 1}{c}{} & 00 & \tikz[baseline=(n.base)] \node[draw,inner sep=2pt] (n) {11}; & 00 & 01 \\
\multicolumn{ 1}{c}{} & 01 &    &    & \tikz[baseline=(n.base)] \node[draw,inner sep=2pt,rounded corners=2pt] (n) {11}; \\
\multicolumn{ 1}{c}{} & 10 &    &    & 10 \\ \hline
\end{tabular}
\caption{Co-factor table representation}
\label{fig:alt-tt}
\end{figure}

In order to better illustrate our algorithm we represent the reversible
function~$f$ by its \emph{co-factor table}.  This table represents~$f$ by means
of the co-factors of the characteristic function~$F$ of~$f$.  For each variable
$x_i$ of the function there are four co-factors
\begin{itemize}
\item $F_{n_i}=(F_{\bar x_i})_{\bar y_i}$ where~$x_i=0$ maps
  to~$y_i=0$ in~$f$,
\item $F_{p'_i}=(F_{x_i})_{\bar y_i}$ where $x_i=1$ maps to~$y_i=0$ in~$f$,
\item $F_{n'_i}=(F_{\bar x_i})_{y_i}$ where~$x_i=0$ maps to~$y_i=1$ in~$f$, and
\item $F_{p_i}=(F_{x_i})_{y_i}$ where~$x_i=1$ maps to~$y_i=1$ in~$f$.
\end{itemize}

\begin{remark}
  Note that this co-factor table representation is used \emph{only} for
  illustrative purposes in this paper.  In the implementation for the algorithm
  the reversible function is stored using BDDs and therefore in a more compact
  way.
\end{remark}

\begin{example}
  The co-factor table of the function in Fig.~\ref{fig:devos-example}(a) is
  depicted in Fig.~\ref{fig:alt-tt} and shows all four co-factors with respect
  to the decomposition of~$x_1$ to~$y_1$ and separates each of the eight
  function entries by its input and output assignments.  Hence, the order of the
  entries in the co-factor table matters since input and output assignments are
  linked according to their position.  As an example, the entry highlighted
  using a rectangle corresponds to the mapping~$101\mapsto 011$ and the entry
  highlighted using a rounded rectangle corresponds to the mapping~$110\mapsto
  111$.

  We can now describe the principle of the algorithm based on the co-factor
  table representation.  The aim of the algorithm is to find the two control
  functions~$l$ and~$r$ that equalize the variables~$x_1$ and~$y_1$.  As a
  result, after applying single-target gates controlled by these functions, no
  pattern remains in the middle part of the co-factor table, i.e.~in
  columns~$F_{p'_1}$ and~$F_{n'_1}$.  In order to understand the \emph{moving}
  of the entries with respect to the control functions, consider the following
  scenario.  If $l(0,1)=1$ and~$r(1,1)=0$ (entry highlighted by an rectangle in
  Fig.~\ref{fig:alt-tt}), the entry \emph{moves} to the column~$F_{n_1}$,
  because the input pattern is inverted at~$x_1$, therefore changing the pattern
  to $001\mapsto011$.
\end{example}

\sn Figure~\ref{fig:movings} shows all possible combinations of the evaluations
of~$l$ and~$r$ and their effect on the entries in~$F_{n_i}$, $F_{p'_i}$,
$F_{n'_i}$, and~$F_{p_i}$ referred to as~$n$, $p'$, $n'$, and~$p$, respectively.

\begin{figure}[b]
\centering
\newcommand{\markinner}[2]{\tikz[remember picture,baseline=(#1.base)] \node[inner sep=0pt,outer sep=0pt] (#1) {#2};}
\begin{tabular}{c|c|c|c|c}
\multicolumn{ 1}{c|}{} & $\bar l(\vec x)\land\bar r(\vec y)$ & $l(\vec x)\land\bar r(\vec y)$ & $\bar l(\vec x)\land r(\vec y)$ & $l(\vec x)\land r(\vec y)$ \\ \hline
\multicolumn{ 1}{c|}{$n$}  & \tikz[baseline=(n.base)] \node[draw,inner sep=2pt] (n) {$n$};  & $p'$ & $n'$ & \tikz[baseline=(n.base)] \node[draw,inner sep=2pt] (n) {$p$};\\
\multicolumn{ 1}{c|}{$p'$} & $p'$ & \markinner{i1}{$n$}  & \markinner{i2}{$p$}  & $n'$ \\
\multicolumn{ 1}{c|}{$n'$} & $n'$ & \markinner{i3}{$p$}  & \markinner{i4}{$n$}  & $p'$ \\
\multicolumn{ 1}{c|}{$p$}  & \tikz[baseline=(n.base)] \node[draw,inner sep=2pt] (n) {$p$};  & $n'$ & $p'$ & \tikz[baseline=(n.base)] \node[draw,inner sep=2pt] (n) {$n$}; \\
\end{tabular}
\begin{tikzpicture}[remember picture,overlay]
  \node[inner sep=1pt,fit=(i1) (i2) (i3) (i4),draw,rounded corners=2pt] {};
\end{tikzpicture}
\caption{Movings of entries with respect to~$l$ and~$r$}
\label{fig:movings}
\end{figure}

\begin{example}
  Hence, when moving the highlighted entry as described in the previous example,
  i.e.~$l(0,1)=1$, also the second pattern in column~$F_{n_1}$,
  i.e.~$001\mapsto001$ will move to column~$F_{p'_1}$ since~$r(0,1)=0$ and each
  input and output subpattern occurs twice in~$f$.
\end{example}

\sn
In order to move all entries to the outer columns, the following two conditions
must hold:
\begin{enumerate}
\item For the outer entries either none or both control functions must evaluate
to true, hence the entry remains at the same position or is transferred
to the other side.
\item For the inner entries exactly one of the control functions must evaluate
to true.
\end{enumerate}

\sn It is not obvious how to determine~$l$ and~$r$ immediately from the given
co-factors.  In the following, we present a method which moves one entry from
the middle blocks to the outside.  Repeating this process until eventually all
entries from the middle blocks have been moved to the outside will guarantee an
equalization of the considered variable.

\begin{example}
  We explain the method based on the co-factor table as previously presented.
  As pattern from the middle blocks we pick~$101\mapsto011$, which is located in
  column~$F_{p'_1}$.  To determine~$l(x_2,x_3)$ and~$r(y_2,y_3)$, we first set
  $l\leftarrow\bar x_2x_3$ according to the assignment to~$x_2$ and~$x_3$
  and~$r\leftarrow\bot$.  This will imply the following steps:

  \begin{itemize}
  \item Because~$l(01)=1$, the highlighted pattern in the column~$F_{p'_1}$ is
    moved to~$F_{n_1}$.
  \item The original mapping~$101 \mapsto011$ changes to~$001 \mapsto 011$,
    because the input at~$x_1$ is inverted.  However, at the same time also the
    pattern~$001\mapsto 001$, currently in column~$F_{n_1}$, would change
    to~$101\mapsto 001$ and be passed to column~$F_{p'}$.  We call this pattern
    the \emph{implied pattern}.
  \item In order to avoid this, one must additionally set~$r\leftarrow\bar
    y_2y_3$ thereby passing the implied pattern over to~$F_{p_1}$ instead, since
    now both~$l$ and~$r$ evaluate to true for this pattern.
\end{itemize}

\begin{figure}[t]
  \centering
  \let\da=\Downarrow
  \begin{tabular}{|>{\bf}c|CCC|>{\quad}L|>{\quad}L|}
\hline
(1) & 101 & \mapsto & 011 & l\gets\bar x_2x_3                     & d=0 \\
    & \da &         &     &                                       &     \\
(2) & 001 & \mapsto & 001 & r\gets\bar y_2y_3                     & d=1 \\
    &     &         & \da &                                       &     \\
(3) & 100 & \mapsto & 101 & l\gets l\lor\bar x_2\bar x_3=\bar x_2 & d=0 \\
    & \da &         &     &                                       &     \\
(4) & 000 & \mapsto & 000 & r\gets r\lor\bar y_2\bar y_3=\bar y_2 & d=1 \\
    &     &         & \da &                                       &     \\
(5) & 011 & \mapsto & 100 &                                       &     \\
\hline
  \end{tabular}
  \caption{Resolving a cycle}
  \label{fig:resolve}
\end{figure}

The assignment of~$r$ will in turn affect another pattern, hence this process is
repeated until the implied pattern is in $F_{n'_1}$.  Figure~\ref{fig:resolve}
illustrates the process for this example.  The initial pattern is given in the
first row and all implied patterns are given in the rows below.  Input patterns
are shown on the left and output patterns are shown on the right.  The
implication of a pattern is indicated by~`$\Downarrow$'.  The last column shows
a Boolean variable~$d$ which indicates whether function~$l$ or~$r$ is updated.
It is used in the algorithm, which is described in the
Sect.~\ref{sec:algorithm}.
\end{example}

\begin{remark}
  In this example only one cycle had to be resolved and therefore the control
  functions~$l$ and~$r$ could directly be determined.  If more than one cycle
  needs to be resolved, the illustrated process is repeated and all obtained
  functions are composed using the `$\oplus$' operation for the overall control
  function.  This procedure is described in Sect.~\ref{sec:algorithm} in more
  detail.
\end{remark}

\section{Characteristic Representation of Reversible Functions}
\label{sec:char-repr-revers}
The algorithm that is proposed in the following section heavily relies on the
characteristic function representation~$F\in\bfunc_{2n}$ of the given reversible
function~$f\in\bfunc_{n,n}$ for which a circuit should be determined.  The BDD
representation of~$F$ allows for efficient function manipulation and evaluation,
which is described in more detail in this section.

\begin{example}
  Figure~\ref{fig:char-bdd} illustrates how the co-factors of a characteristic
  function can be obtained from its BDD representation.  Further, a BDD
  representing the characteristic function of
  \[ f(x_1,x_2)=(x_1\oplus x_2,x_1\land x_2) \]
  is depicted in Fig.~\ref{fig:char-bdd}(b).
\end{example}

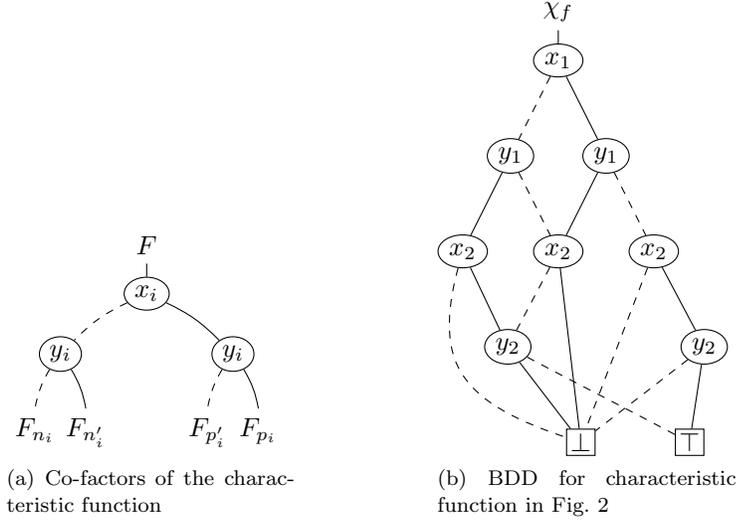
\begin{figure}[t]
  \centering
  \subfloat[Co-factors of the characteristic function]{%
    \begin{tikzpicture}
      \begin{scope}[every node/.style={draw,ellipse,inner sep=1.5pt}]
        \node (x) {$x_i$};
        \node (yl) at (215:1.4cm) {$y_i$};
        \node (yr) at (325:1.4cm) {$y_i$};
      \end{scope}
      \draw[dashed] (x) to[bend right=10] (yl);
      \draw         (x) to[bend left=10] (yr);
      \draw[dashed] (yl) to[bend right=10] +(245:.8cm) node[below] {$F_{n_i}$};
      \draw         (yl) to[bend left=10]  +(295:.8cm) node[below] {$F_{n'_i}$};
      \draw[dashed] (yr) to[bend right=10] +(245:.8cm) node[below] {$F_{p'_i}$};
      \draw         (yr) to[bend left=10]  +(295:.8cm) node[below] {$F_{p_i}$};
      \draw         (x) to                +(90:.4cm)  node[above] {$F$};
    \end{tikzpicture}
  }\hfil
  \subfloat[BDD for characteristic function in Fig.~\ref{fig:truth-table-half-adder}]{%
    \begin{tikzpicture}[scale=.5]
      \begin{scope}[every node/.style={draw,ellipse,inner sep=1.5pt}]
        \node (y11) at (63bp,234bp) [] {$y_1$};
        \node (y21) at (61bp,90bp) [] {$y_2$};
        \node (x11) at (99bp,306bp) [] {$x_1$};
        \node (y22) at (209bp,90bp) [] {$y_2$};
        \node (y12) at (135bp,234bp) [] {$y_1$};
        \node (x21) at (99bp,162bp) [] {$x_2$};
        \node (x23) at (171bp,162bp) [] {$x_2$};
        \node (x22) at (27bp,162bp) [] {$x_2$};
      \end{scope}
      \begin{scope}[every node/.style={draw,inner sep=1.5pt}]
        \node (c0) at (116bp,18bp) [] {$\bot$};
        \node (c1) at (198bp,18bp) [] {$\top$};
      \end{scope}
      \draw [] (y12) ..controls (120.83bp,205.46bp) and (113.11bp,190.44bp)  .. (x21);
      \draw [dashed] (x23) ..controls (154.1bp,117.36bp) and (133.4bp,63.913bp)  .. (c0);
      \draw [] (y11) ..controls (48.835bp,205.46bp) and (41.11bp,190.44bp)  .. (x22);
      \draw [] (x23) ..controls (185.95bp,133.46bp) and (194.11bp,118.44bp)  .. (y22);
      \draw [dashed] (y22) ..controls (176.56bp,64.586bp) and (154.97bp,48.33bp)  .. (c0);
      \draw [] (x11) ..controls (113.17bp,277.46bp) and (120.89bp,262.44bp)  .. (y12);
      \draw [dashed] (y11) ..controls (77.165bp,205.46bp) and (84.89bp,190.44bp)  .. (x21);
      \draw [dashed] (y21) ..controls (105.61bp,66.206bp) and (144.67bp,46.25bp)  .. (c1);
      \draw [dashed] (x22) ..controls (16.907bp,125.43bp) and (12.137bp,94.377bp)  .. (25bp,72bp) .. controls (38.576bp,48.382bp) and (67.176bp,34.303bp)  .. (c0);
      \draw [] (y21) ..controls (81.654bp,62.713bp) and (93.487bp,47.653bp)  .. (c0);
      \draw [] (y22) ..controls (204.63bp,61.211bp) and (202.43bp,47.183bp)  .. (c1);
      \draw [] (x21) ..controls (104.32bp,116.58bp) and (110.67bp,63.52bp)  .. (c0);
      \draw [] (x22) ..controls (40.3bp,133.62bp) and (47.466bp,118.87bp)  .. (y21);
      \draw [dashed] (y12) ..controls (149.17bp,205.46bp) and (156.89bp,190.44bp)  .. (x23);
      \draw [dashed] (x21) ..controls (84.048bp,133.46bp) and (75.894bp,118.44bp)  .. (y21);
      \draw [dashed] (x11) ..controls (84.835bp,277.46bp) and (77.11bp,262.44bp)  .. (y11);
      \draw (x11.north) -- ++(up:10pt) node[above] {$\chi_f$};
    \end{tikzpicture}
  }
  \caption{Characteristic representation of reversible functions}
  \label{fig:char-bdd}
\end{figure}

As pointed out in Lemma~\ref{lem:char-perm-matrix} the characteristic function
of a reversible function corresponds to its Boolean matrix representation.
Since gate application of reversible functions corresponds to matrix
multiplication in the permutation matrix representation,
Lemma~\ref{lem:multiply} tells us that this operation can be carried out
efficiently using BDD manipulation.

Since a reversible function is a 1-to-1 mapping of inputs to outputs, each path
to~\bddTrue~in the BDD for~$F$ visits all variables.  Hence, picking one such
path~$c$ represents one input/output mapping of~$f$ and has a valuation for each
input variable and output variable.  In order to obtain the input pattern
from~$c$ one can remove all outputs using the smoothing operator.  This can be
done analogously to get the output pattern.  In summary we have
\begin{equation}
   \label{eq:smoothing-in-char-f}
   \vec x=\exists\vec y\,c\text{ and }\vec y=\exists\vec x\,c
   \qquad\text{if, and only if}\qquad
   f(\vec x)=\vec y.
\end{equation}

\begin{example}
  If $f\in\bfunc_{3,3}$ is a reversible function with $f(1,0,1)=(1,1,0)$ then
  the pattern~$c=x_1\bar x_2x_3y_1y_2\bar y_3$ is contained in~$F$.  We have
  \[
  \exists\vec y\,c=x_1\bar x_2x_3
  \quad\text{and}\quad
  \exists\vec x\,c=y_1y_2\bar y_3.
  \]
\end{example}

\sn Besides co-factors this operation will be our most powerful tool in order to
describe the synthesis algorithm.  The smoothing operator can also be used
on~$F$ to check whether~$f$ is reversible.  We have: $f$ is reversible, if and
only if
\begin{equation}
  \label{eq:reversibility-char}
  |\onset(F)|=2^n
  \quad\text{and}\quad
  \exists\vec y\,F=\top
  \quad\text{and}\quad
  \exists\vec x\,F=\top,
\end{equation}
i.e.~$F$ has~$2^n$ paths to~\bddTrue, and contains all input and output
patterns.

\section{The Algorithm}
\label{sec:algorithm}
Based on the procedure exemplarily illustrated in the previous section, we first
explain the basic algorithm to decompose a function as
in~\eqref{eq:decomposition} based on a symbolic function representation.
Afterwards, different techniques for optimization are illustrated in this and
successive sections.

\algbegin Algorithm D (Symbolic Decomposition).  Given a reversible
function~$f:\bbbb^n\to\bbbb^n$ that is symbolically represented by its
characteristic function~\[F(\vec x,\vec y)=F(x_1,\dots,x_n,y_1,\dots,y_n)\]
using BDDs and a variable~$x_i\in \{x_1,\dots,x_n\}$, this algorithm finds two
control functions~$l$ and~$r$ from which a reversible
function~$f'=\T[l,i]\circ f\circ\T[r,i]$ can be obtained that does not
change in~$x_i$.

\algstep D1. [Initialization.] Set~$l\leftarrow\bot$ and~$r\leftarrow\bot$.

\algstep D2. [Resolved all cycles?]
If~$F(\vec x,\vec y)\land(x_i\oplus y_i)=\bot$, i.e.~there is no cube in
which~$x_i$ differs from~$y_i$, terminate.  Otherwise, set~$l'\leftarrow\bot$,
$r'\leftarrow\bot$.

\algstep D3. [Pick an inner pattern.] Pick one pattern~$c$ from~$F(\vec x,\vec
y)\land x_i\land\bar y_i$, i.e.~$c$ represents a mapping in which~$x_i=1$ is
mapped to~$y_i=0$.  Also, set the \emph{direction bit}~$d\gets 0$.

\algstep D4. [Update~$l'$ or~$r'$.] Whether~$l'$ or~$r'$ is updated depends
on the value of~$d$: \\
1)~If~$d=0$, let~$c'=\exists x_i\,\exists\vec y\,c$ and set~$l'\gets l'\lor c'$
and~$c\gets F(\vec x,\vec y)\land\bar x_i\land c'$. \\
2)~If~$d=1$,
let~$c'=\exists y_i\,\exists\vec x\,c$ and set~$r'\gets r'\lor c'$
and~$c\gets F(\vec x,\vec y)\land y_i\land c'$.

\algstep D5. [Change direction.] Set~$d\leftarrow 1-d$.

\algstep D6. [Repeat iteration?] If~$F(\vec x,\vec y)\land c\land(x_i\oplus
y_i)=\bot$, i.e.~the newly derived cube~$c$ represents a mapping in which~$x_i$
equals~$y_i$, return to step~4.  Otherwise, continue.

\algstep D7. [Update gate functions.]
  Update~$l$ and~$r$ using~$l'$ and~$r'$ and recompute~$f$.
  Return to step~2.\quad\slug

  \sn The algorithm creates the control functions~$l$ and~$r$ in an iterative
  manner.  Being initially assigned~$\bot$ in step~1, in each iteration of the
  outer loop (steps~2 to~7), the algorithm resolves one cycle as described in
  Fig.~\ref{fig:resolve} and stores the resulting functions in~$l'$ and~$r'$.
  To resolve a cycle first one input/output pattern~$c$ is picked in
  which~$x_i=1$ and~$y_i=0$~(step~3).  Based on this pattern successive patterns
  are implied as described in the previous section and~$l'$ and~$r'$ are updated
  accordingly~(steps~4 to~6).  This inner loop is repeated as long as the
  implied pattern is neither contained in~$F_{n'_i}$ nor in~$F_{p'_i}$,
  i.e.~\[F(\vec x,\vec y)\land c\land(x_i\oplus y_i)=\bot.\] Eventually, after a
  cycle has been resolved and~$l'$ and~$r'$ are determined, the control
  functions~$l$ and~$r$ are updated and~$f$ is recomputed.  For this purpose, we
  first set~\[r'\leftarrow \exists \vec y\,\left(r'\land
    \bigwedge_{j=1}^n(x_j\leftrightarrow y_j)\right).\] This changes~$r'$ such
  that all variables~$y_j$ are substituted by corresponding~$x_j$.  Afterwards,
  we set~$l\leftarrow l\oplus l'$,~$r\leftarrow r\oplus r'$, and
\begin{equation}
  \label{eq:update-f}
  f\leftarrow\T[l',i]\circ f\circ\T[r',i].
\end{equation}

Note that throughout the whole procedure of Algorithm~D the reversible
function will only be represented as its characteristic function using BDDs and
all operations will be carried out on this symbolic representation.  This
includes~\eqref{eq:update-f} for which also the single-target functions are
represented by its characteristic functions.  Function composition then
corresponds to Boolean matrix multiplication which can efficiently be
implemented using the smoothing operator~\citep{Knuth:2011}.

\begin{theorem}
  Algorithm~D is sound and complete.
\end{theorem}

\begin{proof}
  The inner loop will always terminate as illustrated in Fig.~\ref{fig:resolve}.
  Also, by resolving a cycle Algorithm~D will always move at least one pattern
  from the \emph{inner} co-factors~$F_{p'_i}$ and~$F_{n'_i}$ to the \emph{outer}
  co-factors~$F_{n_i}$ and~$F_{p_i}$.  Hence, the termination condition for the
  outer loop in step~2 will eventually hold.  The algorithm is sound since the
  input function is updated using the same extracted gates that are returned.
\end{proof}

\section{Increasing Efficiency}
\label{sec:incr-effic}
This section describes two strategies to increase the efficiency of Algorithm~D
by considering certain special cases explicitly.
\subsection{Only one single-target gate required}
It turns out that in many non-random functions only one single-target gate is
sufficient in order to equalize a variable, i.e~either~$l=\bot$ or~$r=\bot$.
Also, this case can easily be checked and if it applies, the control functions
can efficiently be computed.

We consider the case in which~$r$ may be~$\bot$ and hence only one single-target
gate needs to be added for the control function~$l$.  The check first removes
the variable~$x_i$ from the characteristic function and then reassigns~$x_i$
such that it equals~$y_i$.  The result is assigned to~$F'$:
\begin{equation}
  F'=\exists x_i\,F\land(x_i\leftrightarrow y_i)
\end{equation}
If~$F'$ is reversible, only one single-target gate is required.  Reversibility
of~$F'$ can be checked by the condition
\begin{equation}
  \label{eq:reversibility-check}
  \exists\vec x\, F'=\top,
\end{equation}
i.e.~$F'$ still consists of all possible input patterns.  The operation
in~\eqref{eq:reversibility-check} can efficiently be performed using BDDs.

\begin{figure}
  \def\tabcolsep{1.15pt}
  \def\arraystretch{1}
  \subfloat[Negative example]{%
    \begin{tabular}{|CCC|CCC|} \hline
      x_1&x_2&x_3&y_1&y_2&y_3 \\ \hline
      0 & 0 & 0 & 0 & 0 & 0 \\
      0 & 0 & 1 & 0 & 0 & 1 \\
      0 & 1 & 0 & 0 & 1 & 0 \\
      0 & 1 & 1 & 1 & 0 & 0 \\
      1 & 0 & 0 & 1 & 0 & 1 \\
      1 & 0 & 1 & 0 & 1 & 1 \\
      1 & 1 & 0 & 1 & 1 & 1 \\
      1 & 1 & 1 & 1 & 1 & 0 \\ \hline
    \end{tabular}
    \begin{tabular}{|>{\bf}CCC|CCC|} \hline
      \text{$x_1$}&x_2&x_3&y_1&y_2&y_3 \\ \hline
      0 & 0 & 0 & 0 & 0 & 0 \\
      0 & 0 & 1 & 0 & 0 & 1 \\
      0 & 1 & 0 & 0 & 1 & 0 \\
      1 & 1 & 1 & 1 & 0 & 0 \\
      1 & 0 & 0 & 1 & 0 & 1 \\
      0 & 0 & 1 & 0 & 1 & 1 \\
      1 & 1 & 0 & 1 & 1 & 1 \\
      1 & 1 & 1 & 1 & 1 & 0 \\ \hline
    \end{tabular}
  }\hfill
  \subfloat[Positive example]{%
    \begin{tabular}{|CCC|CCC|} \hline
      x_1&x_2&x_3&y_1&y_2&y_3 \\ \hline
      0 & 0 & 0 & 0 & 0 & 1 \\
      0 & 0 & 1 & 0 & 0 & 0 \\
      0 & 1 & 0 & 0 & 1 & 1 \\
      0 & 1 & 1 & 1 & 0 & 0 \\
      1 & 0 & 0 & 1 & 1 & 1 \\
      1 & 0 & 1 & 1 & 1 & 0 \\
      1 & 1 & 0 & 1 & 0 & 1 \\
      1 & 1 & 1 & 0 & 1 & 0 \\ \hline
    \end{tabular}
    \begin{tabular}{|>{\bf}CCC|CCC|} \hline
      \text{$x_1$}&x_2&x_3&y_1&y_2&y_3 \\ \hline
      0 & 0 & 0 & 0 & 0 & 1 \\
      0 & 0 & 1 & 0 & 0 & 0 \\
      0 & 1 & 0 & 0 & 1 & 1 \\
      1 & 1 & 1 & 1 & 0 & 0 \\
      1 & 0 & 0 & 1 & 1 & 1 \\
      1 & 0 & 1 & 1 & 1 & 0 \\
      1 & 1 & 0 & 1 & 0 & 1 \\
      0 & 1 & 1 & 0 & 1 & 0 \\ \hline
    \end{tabular}
  }
  \caption{Negative and positive example for checking whether one single-target
    gate is sufficient.  The original truth table is always on the left-hand
    side and an updated version in which~$x_1$ is set equal to~$y_1$ on the
    right-hand side.}
  \label{fig:only-one-gate}
\end{figure}

\begin{example}
  A negative and positive example of this optimization step is illustrated in
  Fig.~\ref{fig:only-one-gate}.  In both cases the variable~$x_1$ is updated to
  equal the value of~$y_1$.  Figure~\ref{fig:only-one-gate}(a) shows the truth
  table from Fig.~\ref{fig:devos-example}(a) and it can be seen that the
  resulting truth table represents an irreversible function since the input
  patterns~$001$ and~$111$ occur twice.  For the function in
  Fig.~\ref{fig:devos-example}(b) one single-target gate is sufficient to
  equalize~$x_1$ and~$y_1$ since the resulting truth table represents a
  reversible function.
\end{example}

\sn Once $F'$ has been determined one can determine the control function~$l$ by
inspecting the input patterns which differ from the original function~$F$.  For
this purpose one needs to remove all output variables and~$x_i$ from their
intersection, i.e.
\begin{equation}
  \label{eq:new-control-function-special-case}
  \exists x_i\,\exists\vec y\,(F\land F').
\end{equation}

\begin{example}
  For the function in Fig.~\ref{fig:only-one-gate} the application
  of~\eqref{eq:new-control-function-special-case} yields
  \[ \exists x_1\,\exists\vec y\,(F\land F') = x_2x_3. \]
\end{example}

\sn
In order to consider this special case, we added a preprocessing step to
Algorithm~D before step~2:

\algindentset{D1a}
\algstep D1a. [Is $r=\bot$?] Let~$F'=\exists x_i\, F
\land (x_i\leftrightarrow y_i)$.
If~$\exists\vec x\, F'=\top$,
set~$l\leftarrow\exists x_i\,\exists\vec y\,(F\land F')$ and terminate.

\sn
The same can be done for the case that~$l=\bot$ analogously:

\algstep D1b. [Is $l=\bot$?] Let~$F'=\exists y_i\, F
\land (x_i\leftrightarrow y_i)$.
If~$\exists\vec y\, F'=\top$,
set~$r\leftarrow\exists y_i\,\exists\vec x\,(F\land F')$ and terminate.

\subsection{Cycles of length~$1$ and~$2$}
To resolve a cycle in Algorithm~D one computes implied patterns iteratively in
each step.  As a result, cycles need to be resolved one after the other, which
becomes time consuming if many cycles need to be resolved. Alternatively, one
can also express cycles of a fixed length in one Boolean formula and resolve
them all at once.  However, the size of the formula grows with respect to the
length of the cycles and is therefore only applicable to small lengths.  In our
implementation we considered cycles of length~$1$ and~$2$.

\begin{example}
  The columns~$F_{p'_1}$ and~$F_{n'_1}$ of the co-factor table in
  Fig.~\ref{fig:length-optimization} both contain the entries~$01$ and~$10$.  A
  single-target gate controlled by~$l'=x_2\oplus x_3$ can resolve both of them
  at once without the computationally more expensive iteration steps in
  Algorithm~D.
\end{example}

\begin{figure}[t]
  \def\tabcolsep{1.15pt}
  \def\arraystretch{1}
\centering
\subfloat[Length~$1$]{%
\begin{tabular}{ccccc}
\hline
\multicolumn{1}{c}{} & $F_{n_1}$ & $F_{p'_1}$ & $F_{n'_1}$ & $F_{p_1}$ \\
\multicolumn{1}{c}{$x_1 \mapsto y_1$ } & $0 \mapsto 0$ & $1 \mapsto 0$  & $0 \mapsto 1$ & $1 \mapsto 1$  \\ \hline
\parbox[t]{2mm}{\multirow{4}{*}{\rotatebox[origin=c]{90}{\footnotesize Inputs}}} & $x_2x_3$  & $x_2 x_3$  & $x_2 x_3$ & $x_2 x_3$ \\ \cline{ 2- 5}\noalign{\vspace{2pt}}
\multicolumn{ 1}{c}{} & 11 & 01 & 01 & 11 \\
\multicolumn{ 1}{c}{} &    & 10 & 10 &    \\
\multicolumn{ 1}{c}{} &    & 11 & 00 &    \\  \hline
\parbox[t]{2mm}{\multirow{4}{*}{\rotatebox[origin=c]{90}{\footnotesize Outputs}}} & $y_2 y_3$  & $y_2 y_3$  & $y_2 y_3$  & $y_2 y_3$  \\ \cline{ 2- 5}\noalign{\vspace{2pt}}
\multicolumn{ 1}{c}{} & 00 & 11 & 00 & 01 \\
\multicolumn{ 1}{c}{} &    & 10 & 11 &    \\
\multicolumn{ 1}{c}{} &    & 01 & 10 &    \\ \hline
\end{tabular}}\hfill
\subfloat[Length~$2$]{%
\begin{tabular}{ccccc}
\hline
\multicolumn{1}{c}{} & $F_{n_1}$ & $F_{p'_1}$ & $F_{n'_1}$ & $F_{p_1}$ \\
\multicolumn{1}{c}{$x_1 \mapsto y_1$ } & $0 \mapsto 0$ & $1 \mapsto 0$  & $0 \mapsto 1$ & $1 \mapsto 1$  \\ \hline
\parbox[t]{2mm}{\multirow{4}{*}{\rotatebox[origin=c]{90}{Inputs}}} & $x_2x_3$  & $x_2 x_3$  & $x_2 x_3$ & $x_2 x_3$ \\ \cline{ 2- 5}\noalign{\vspace{2pt}}
\multicolumn{ 1}{c}{} & 00 & 01 & 11 & 00 \\
\multicolumn{ 1}{c}{} & 01 &    &    & 10 \\
\multicolumn{ 1}{c}{} & 10 &    &    & 11 \\ \hline
\parbox[t]{2mm}{\multirow{4}{*}{\rotatebox[origin=c]{90}{Outputs}}} & $y_2 y_3$  & $y_2 y_3$  & $y_2 y_3$  & $y_2 y_3$  \\ \cline{ 2- 5}\noalign{\vspace{2pt}}
\multicolumn{ 1}{c}{} & 00 & 11 & 10 & 00 \\
\multicolumn{ 1}{c}{} & 10 &    &    & 11 \\
\multicolumn{ 1}{c}{} & 01 &    &    & 01 \\ \hline
\end{tabular}}
\caption{Cycles of length~$1$ and~$2$}
\label{fig:length-optimization}
\end{figure}

\sn In general such control functions for cycles of length~$1$ are expressed by
\begin{equation}
  \label{eq:cycles-of-length-1}
  l'\gets\exists\vec y\,F_{p'_i}\land\exists\vec y\,F_{n'_i}
  \quad\text{and}\quad
  r'\gets\exists\vec x\,F_{p'_i}\land\exists\vec x\,F_{n'_i},
\end{equation}
i.e.~the intersection of the co-factors is computed after dropping the outputs
or inputs, respectively.  It is important not to compute these two control
functions at the same time as they may depend on each other.  Furthermore, in
our experimental evaluation we discovered that the best results are obtained
when~$f$ gets recomputed after each control function has been determined.
Cycles of length~$1$ are therefore handled using the following additional step:

\algstep D1c. [$1$-cycles.] Set~$l'\leftarrow\exists\vec y\,F_{p'_i}\land\exists
\vec y\,F_{n'_i}$.
Update~$l$ using~$l'$ and recompute~$f$.
Set~$l'\leftarrow\bot$ and set~$r'\leftarrow\exists\vec x\,F_{p'_i}\land\exists
\vec x\,F_{n'_i}$.
Update~$r$ using~$r'$ and recompute~$f$.

\sn The `updates' correspond to the same procedure as referred to in step~7 in
Algorithm~D.  This idea can be extended to cycles of length~$2$, as illustrated
by means of the following example.

\begin{example}
  Consider the co-factor table in Fig.~\ref{fig:length-optimization}(b).
  Starting with the pattern~$101\mapsto 011$ (column~$F_{p'_1}$), Algorithm~D
  will first imply the pattern~$001\mapsto 010$ (column~$F_{n_1}$) and
  afterwards pattern~$011\mapsto 110$ (column~$F_{n'_1}$).  Then, the cycle is
  resolved.  In these three patterns the input part~$01$ of the middle
  pattern~$001\mapsto 010$ matches the input part of the first pattern that is
  contained in column~$F_{p'_1}$.  At the same time, the output part~$10$
  matches the output part of the last pattern that is contained in
  column~$F_{n'_1}$ and therefore $l'=\bar x_2x_3$ and~$r'=y_2\bar y_3$.
\end{example}

\sn In general, cycles of length~$2$ are obtained by
\begin{equation}
  \label{eq:cycles-of-length-2}
  l'\leftarrow\exists\vec y\,g
  \quad\text{and}\quad
  r'\leftarrow\exists\vec x\,g,
  \quad\text{where $g=F_{n_i}\land\exists\vec y\,F_{p'_i}
\land\exists\,\vec x F_{n'_i}$}.
\end{equation}

\sn Cycles of length~$2$ are therefore handled using the following additional
step:

\algstep D1d. [$2$-cycles.]  Let~$g=F_{n_i}\land\exists\vec y\,F_{p'_i}
\land\exists\,\vec x F_{n'_i}$. Set~$l'\leftarrow\exists\vec y\,g$ and $r'
\leftarrow\exists\vec x\,g$.  Update~$l$ and~$r$ using~$l'$ and~$r'$ and
recompute~$f$.  Let~$g=F_{p_i}\land\exists\vec y\,F_{n'_i} \land\exists\,\vec x
F_{p'_i}$. Set~$l'\leftarrow\exists\vec y\,g$ and $r' \leftarrow\exists\vec
x\,g$.  Update~$l$ and~$r$ using~$l'$ and~$r'$ and recompute~$f$.

\section{Variable Ordering}
\label{sec:ordering}
In order to obtain a circuit from a given reversible
function~$f:\bbbb^n\to\bbbb^n$, one needs to apply Algorithm~D for each of the
$n$ variables.  As a result, Algorithm~D is called~$n$ times.  However, the
order in which variables are equalized by Algorithm~D does not necessarily need
to be the natural variable ordering~$x_1,x_2,\dots,x_n$, which is e.g.~being used
in~\citep{VR:2008}.  Instead, any of the possible~$n!$ different variable
orderings can be used and each ordering may lead to a different circuit
realization with respect to the number of gates.  Since the number of different
variable orderings is large, finding the smallest circuit cannot be done
efficiently.  In this section we are motivating two heuristics that may lead to
cheaper circuits.

\subsection{Greedy Heuristic}
\label{sec:greedy}
The first heuristic follows a greedy approach.  In each step the two
single-target gates for all variables that have not yet been processed are
computed.  The single-target gate-pair is chosen that leads to the
cheapest realization after expansion to Toffoli gates.  Using this heuristic the
algorithm is executed~$n(n+1)/2$ times for a function with~$n$ variables.
Consequently, using this heuristic will probably increase the overall run-time.

\subsection{Hamming Distance Heuristic}
\label{sec:hamming}
Algorithm~D computes control functions~$l$ and~$r$ by equalizing the values for
the variables~$x_i$ and~$y_i$.  Consequently, it seems plausible to use
variables~$x_i$ and~$y_i$ that already agree on many of the input/output
mappings.  This heuristic first counts this number for each variable and then
chooses the one with the maximum number of such patterns.

Note that the patterns do not have to be counted explicitly.  The number of such
patterns is equal to
\begin{equation}
  \label{eq:hamming-distance-heuiristic}
  |\onset(F(\vec x,\vec y)\land(x_i\oplus y_i))|.
\end{equation}
Obviously the number of paths in the co-factors~$F_{p'_i}$ and~$F_{n'_i}$ are
small when the result of~\eqref{eq:hamming-distance-heuiristic} is large.  The
assumption is that in this case less cycles need to be resolved.  However, the
experimental results will not confirm this assumption.

\section{Partial Functions}
\label{sec:partial-functions}
\begin{figure}
  \centering
  \def\tabcolsep{1.15pt}
  \def\arraystretch{1}
  \begin{tabular}{|CCC|CCC|} \hline
    x_1&x_2&x_3&y_1&y_2&y_3 \\ \hline
    0 & 0 & 0 & 0 & 1 & 0 \\
    0 & 1 & 0 & 1 & 0 & 0 \\
    0 & 1 & 1 & 0 & 1 & 1 \\
    1 & 0 & 0 & 0 & 0 & 0 \\
    1 & 0 & 1 & 1 & 1 & 0 \\ \hline
  \end{tabular}
  \caption{Partial function represented as truth table}
  \label{fig:partial-function}
\end{figure}

A partial reversible function on~$n$ variables is a function that does not
contain all input/output mappings.  The original algorithm presented
in~\citep{VR:2008} cannot be applied to such functions in its original form
since then cycles may not be complete.  Algorithm~D can however be extended in
order to support partial functions by extending the function on demand.

\begin{example}
  A partial reversible function is given in terms of its truth table
  representation in Fig.~\ref{fig:partial-function}.  Its representation as
  characteristic function is
  \[
  \begin{aligned}
    F(x_1,x_2,x_3,y_1,y_2,y_3) &
    = \bar x_1\bar x_2\bar x_3\bar y_1y_2\bar y_3
    \lor\bar x_1x_2\bar x_3y_1\bar y_2\bar y_3
    \lor\bar x_1x_2x_3\bar y_1y_2y_3 \\
    & \,\lor x_1\bar x_2\bar x_3\bar y_1\bar y_2\bar y_3
    \lor x_1\bar x_2x_3y_1y_2\bar y_3.
  \end{aligned}
  \]
\end{example}

\sn
When dealing with partial functions, Algorithm~D may run into the problem that
the implied pattern~$c$ may be assigned~$\bot$ after step~4 when a corresponding
pattern is not defined by the original function.

\begin{figure}[t]
  \def\tabcolsep{1.15pt}
  \def\arraystretch{1}
\centering
\subfloat[Co-factor table]{%
\begin{tabular}{ccccc}
\hline
\multicolumn{1}{c}{} & $F_{n_1}$ & $F_{p'_1}$ & $F_{n'_1}$ & $F_{p_1}$ \\
\multicolumn{1}{c}{$x_1 \mapsto y_1$ } & $0 \mapsto 0$ & $1 \mapsto 0$  & $0 \mapsto 1$ & $1 \mapsto 1$  \\ \hline
\parbox[t]{2mm}{\multirow{4}{*}{\rotatebox[origin=c]{90}{\footnotesize\enspace\enspace Inputs}}} & $x_2x_3$  & $x_2 x_3$  & $x_2 x_3$ & $x_2 x_3$ \\ \cline{ 2- 5}\noalign{\vspace{2pt}}
\multicolumn{ 1}{c}{} & 00 & 00 & 10 & 01 \\
\multicolumn{ 1}{c}{} & 11 &    &    &    \\ \hline
\parbox[t]{2mm}{\multirow{4}{*}{\rotatebox[origin=c]{90}{\footnotesize\enspace\enspace Outputs}}} & $y_2 y_3$  & $y_2 y_3$  & $y_2 y_3$  & $y_2 y_3$  \\ \cline{ 2- 5}\noalign{\vspace{2pt}}
\multicolumn{ 1}{c}{} & 10 & 00 & 00 & 10 \\
\multicolumn{ 1}{c}{} & 11 &    &    &    \\ \hline
\end{tabular}}\hfill\
\subfloat[Pattern implication]{%
  \let\da=\Downarrow
  \begin{tabular}{|>{\bf}c|CCC|>{\quad}L|}
\hline
(1) & 100 & \mapsto & 000 & l'\gets\bar x_2\bar x_3            \\
    & \da &         &     &                                    \\
(2) & 000 & \mapsto & 010 & r'\gets y_2\bar y_3                \\
    &     &         & \da &                                    \\
(3) & 101 & \mapsto & 110 & l'\gets l'\lor\bar x_2x_3=\bar x_2 \\
    & \da &         &     &                                    \\
(4) & ?   &         &     &                                    \\
\hline
  \end{tabular}
}
\caption{Resolving a cube in a partial function}
\label{fig:partial-functions-algorithm-d}
\end{figure}

\begin{example}
  \label{ex:resolving-in-partial}
  Figure~\ref{fig:partial-functions-algorithm-d} demonstrates this problem.  In
  Fig.~\ref{fig:partial-functions-algorithm-d}(a) the co-factor table of the
  function in Fig.~\ref{fig:partial-function} is shown to aid the comprehension
  of the resolving steps in Fig.~\ref{fig:partial-functions-algorithm-d}(b).  As
  can be seen, two further patterns are applied from the initial one that have
  been picked from the co-factor~$F_{p'_1}$.  The iteration step is not
  completed after the third pattern, however, there is no pattern available to
  continue that has the input~$001$.
\end{example}

\sn The idea is to extend the function in such situations.  All not specified
patterns can be chosen arbitrarily as long as reversibility of~$f$ is ensured.
For illustration purposes, let us consider the case in which after step~4 we
have $c=\bot$ and~$d=0$.  That is, there is no pattern specified in~$f$ for the
input pattern~$\bar x_ic'$.  It is sufficient to set~$c\leftarrow\bar x_ic'c''$
where~$c''$ is a cube randomly picked from~$\exists\vec x\,\bar F$, i.e.~all output
patterns that are not specified in~$f$.  Afterwards, the new pattern is added
to~$f$ by setting~$f\leftarrow f\lor c$.  Consequently, it is not contained
in~$\exists\vec x\,\bar F$ any longer.

\begin{example}
  For Example~\ref{ex:resolving-in-partial} we have~$x_i=x_1$, $c'=\bar
  x_2x_3$, and
  \[
  \exists\vec x\,\bar F=\bar y_1\bar y_2y_3\lor y_1\bar y_2y_3\lor y_1y_2y_3.
  \]
\end{example}

\sn In order to further increase the efficiency, our implementation first tries
to find an output pattern~$c''$ that is in~$y_i\land \exists\vec x\,\bar F$,
i.e.~patterns in which~$y_i$ is assigned~$1$.  In this case, one can ensure that
the condition in step~6 fails and the cycle has been resolved immediately.

\begin{remark}
  Note that of course both algorithms can deal with partial functions if they
  are made total prior to synthesis by randomly filling all non-specified
  mappings.  However, that technique may not always be efficient.  In contrast,
  our proposed technique does not necessarily add all non-specified mappings.
\end{remark}

\newcolumntype{R}[1]{>{\raggedright\arraybackslash }b{#1}}

\section{Experimental Evaluations}
\label{sec:exper-eval}
The synthesis approach we propose has been implemented in C++ using BDDs on top
of RevKit~\citep{SFWD:12} and CUDD~\citep{Somenzi:2001}.\footnote{The source
  code that has been used to perform this evaluation is available at
  www.revkit.org (version 2.0).}  The experimental evaluation has been carried out on a 3.4~GHz
Quad-Core Intel Xeon Processor with 32~GB of main memory running Linux~3.14.
The experimental results in Sects.~\ref{sec:comp-relat-algor}
and~\ref{sec:eval-vari-order} have been generated with the RevKit program
`\emph{rcbdd\_synthesis}'.
For all experiments, we have set the time-out to 10000~seconds.  We have
verified all our results using the tool
\texttt{abc}~\citep{BM:10}.\footnote{The tool \texttt{abc} can be downloaded from
  bitbucket.org/alanmi/abc. The read routine for reversible circuit files is
  available at bitbucket.org/msoeken/abc} The following sections discuss the
experiments.

\begin{table}[t]%
\fontsize{7pt}{8pt}\selectfont
\caption{Comparison to related algorithms}
\label{tbl:other-approaches}
\begin{center}
\renewcommand{\tabcolsep}{2pt}
\begin{tabularx}{\linewidth}{Xr|rrr|rrr||rrrrr}
\hline 

 & \multicolumn{1}{r|}{} & \multicolumn{3}{c|}{Truth table based } &\multicolumn{3}{c||}{QMDD based} &  \multicolumn{5}{c}{BDD based}  \\
 & \multicolumn{1}{r|}{} & \multicolumn{3}{c|}{(De Vos et al., 2008)}
 &\multicolumn{3}{c||}{\citep{SWH+:2012}}
 &  \multicolumn{5}{c}{Proposed approach~(Sect.~\ref{sec:algorithm})}  \\ 

Name & $n$ & $d_{\rm T}$ & $q_{\rm T}$ & $t_{\rm T}$ & $d_{\rm Q}$ & $q_{\rm Q}$
& $t_{\rm Q}$ & $d$ & $q$ & $t$ & Imp.~T & Imp.~Q \\ \hline
sym6       & 7  & 172   & 12852    & 0.48    & 262   & 41552    & \bf0.02    & 124   & 8911     & 0.11       & 30.66    & 78.55    \\
urf2       & 8  & 279   & 26257    & 0.15    & 763   & 165697   & \bf0.14    & 268   & 24066    & 0.19       & 8.34     & 85.48    \\
con1       & 8  & 327   & 32102    & 2.06    & 659   & 139118   & \bf0.09    & 233   & 22988    & 0.23       & 28.39    & 83.48    \\
hwb9       & 9  & 797   & 107219   & 9.26    & 2294  & 629433   & 0.91       & 584   & 73465    & \bf0.79    & 31.48    & 88.33    \\
urf1       & 9  & 812   & 106176   & 6.74    & 1957  & 533680   & 0.71       & 563   & 74858    & \bf0.66    & 29.50    & 85.97    \\
urf5       & 9  & 388   & 61649    & 2.78    & 693   & 185752   & \bf0.13    & 213   & 32676    & 0.16       & 47.00    & 82.41    \\
adr4       & 9  & 681   & 101577   & 7.16    & 1130  & 290997   & \bf0.26    & 459   & 64309    & 0.78       & 36.69    & 77.90    \\
sym9       & 10 & 1653  & 255990   & 15.34   & 3895  & 1276583  & \bf2.94    & 1175  & 174678   & 3.19       & 31.76    & 86.32    \\
urf3       & 10 & 1690  & 268828   & 16.86   & 4071  & 1347318  & 3.76       & 1081  & 162225   & \bf2.15    & 39.65    & 87.96    \\
5xp1       & 10 & 1507  & 232995   & 11.19   & 2231  & 724052   & \bf1.07    & 837   & 134267   & 3.99       & 42.37    & 81.46    \\
rd84       & 11 & 3641  & 689283   & 31.53   & 6389  & 2445443  & \bf10.95   & 2063  & 401660   & 20.60      & 41.73    & 83.58    \\
sym10      & 11 & 3657  & 664202   & 33.31   & 9812  & 3780469  & 23.05      & 2467  & 461538   & \bf12.01   & 30.51    & 87.79    \\
urf4       & 11 & 3911  & 713328   & 37.11   & 11684 & 4508910  & 40.74      & 2641  & 491645   & \bf12.10   & 31.08    & 89.10    \\
clip       & 11 & 3542  & 685118   & 35.55   & 7913  & 3036313  & \bf16.88   & 2271  & 434952   & 18.06      & 36.51    & 85.67    \\
cycle10\_2 & 12 & 27    & 4200     & 8.41    & 36    & 6286     & 0.07       & 27    & 4200     & \bf0.05    & 0.00     & 33.18    \\
dc2        & 13 & 7999  &  2706179 & 119.73  & 11346 &  5612922 & \bf74.78   & 4102  &  1395422 & 224.06     & 48.44    & 75.14    \\
misex1     & 14 & 34671 & 12014870 & 539.30  & 18412 & 10115630 & \bf274.82  & 6867  & 2733073  & 1046.12    & 77.25    & 72.98    \\
co14       & 15 & 83652 & 30502311 & 2032.06 & 63640 & 38678808 & 2677.82    & 25065 & 10028634 & \bf730.96  & 67.12    & 74.07    \\
urf6       & 15 & 15197 & 7275366  & 2855.35 & 23497 & 14432936 & 336.52     & 2164  & 1215312  & \bf3.96    & 83.30    & 91.58    \\
dk27       & 15 & 50807 & 19144930 & 2074.22 & 57619 & 35018963 & \bf3773.21 & 23882 & 11254565 & 8598.83    & 41.21    & 67.86    \\
C7552      & 20 & ---   & ---      & TO      & 356   & 309008   & 133.89     & 257   & 180894   & \bf8.53    & ---      & 41.46    \\
bw         & 32 & ---   & ---      & TO      & ---   & ---      & MO         & 2585  & 3766784  & \bf2076.51 & ---      & ---      \\
\hline \multicolumn{13}{p{.99\linewidth}}{%
  \vspace{0pt} The percentage improvements are with respect to the values for
  quantum costs in the blocks \emph{Truth table based} and \emph{QMDD based}.  }
\end{tabularx}
\end{center}
\end{table}

\begin{table}[t]
\fontsize{7pt}{8pt}\selectfont
\caption{Evaluation of variable heuristics}
\label{tbl:orderings}
\begin{center}
\renewcommand{\tabcolsep}{2.5pt}
\begin{tabularx}{\linewidth}{Xr|rrr||rrrr|rrrr}
\hline 
   & \multicolumn{1}{r|}{} & \multicolumn{3}{c||}{BDD based~(Sect.~\ref{sec:algorithm})} &  \multicolumn{4}{c|}{Greedy~(Sect.~\ref{sec:greedy})} & \multicolumn{4}{c}{Hamming~(Sect.~\ref{sec:hamming})}   \\

 Name &  $n$ &  $d$ & $q$ & $t$ & $d$ & $q$ & $t$ & Imp. & $d$ & $q$ & $t$ & Imp. \\  \hline 
sym6       & 7  & 124   & 8911      & 0.11    & 113  & 8540       & 0.56     & 4.16     & 107  & \bf7973    & 0.11   & 10.53    \\
urf2       & 8  & 268   & \bf24066  & 0.19    & 259  & 24528      & 1.59     & -1.92    & 263  & 24535      & 0.27   & -1.95    \\
con1       & 8  & 233   & 22988     & 0.23    & 201  & \bf20104   & 1.58     & 12.55    & 234  & 20447      & 0.34   & 11.05    \\
hwb9       & 9  & 584   & \bf73465  & 0.79    & 572  & 74690      & 8.29     & -1.67    & 584  & \bf73465   & 0.96   & 0.00     \\
urf1       & 9  & 563   & 74858     & 0.66    & 556  & 74116      & 6.79     & 0.99     & 573  & \bf73094   & 0.92   & 2.36     \\
urf5       & 9  & 213   & \bf32676  & 0.16    & 188  & 33243      & 2.11     & -1.74    & 307  & 46053      & 0.34   & -40.94   \\
adr4       & 9  & 459   & 64309     & 0.78    & 171  & \bf23401   & 2.93     & 63.61    & 195  & 26537      & 0.42   & 58.74    \\
sym9       & 10 & 1175  & 174678    & 3.19    & 814  & \bf148295  & 27.45    & 15.10    & 1014 & 164388     & 3.95   & 5.89     \\
urf3       & 10 & 1081  & \bf162225 & 2.15    & 1051 & 162267     & 29.27    & -0.03    & 1058 & 165242     & 2.92   & -1.86    \\
5xp1       & 10 & 837   & \bf134267 & 3.99    & 838  & 138159     & 26.73    & -2.90    & 847  & 147315     & 5.38   & -9.72    \\
rd84       & 11 & 2063  & 401660    & 20.60   & 1996 & \bf392350  & 178.86   & 2.32     & 1868 & 401940     & 26.01  & -0.07    \\
sym10      & 11 & 2467  & 461538    & 12.01   & 1405 & \bf334313  & 123.13   & 27.57    & 2068 & 437864     & 17.63  & 5.13     \\
urf4       & 11 & 2641  & 491645    & 12.10   & 2629 & 496356     & 187.30   & -0.96    & 2629 & \bf490826  & 22.21  & 0.17     \\
clip       & 11 & 2271  & 434952    & 18.06   & 2214 & \bf423724  & 193.36   & 2.58     & 2344 & 455392     & 32.86  & -4.70    \\
cycle10\_2 & 12 & 27    & \bf4200   & 0.05    & 27   & \bf4200    & 0.39     & 0.00     & 27   & \bf4200    & 0.05   & 0.00     \\
dc2        & 13 & 4102  & 1395422   & 224.06  & 3645 & 1274742    & 1131.20  & 8.65     & 3230 & \bf1190154 & 185.02 & 14.71    \\
misex1     & 14 & 6867  & 2733073   & 1046.12 & 7937 & 3111290    & 8044.99  & -13.84   & 3949 & \bf1714174 & 480.96 & 37.28    \\
co14       & 15 & 25065 & 10028634  & 730.96  & 1603 & \bf714287  & 671.89   & 92.88    & ---  & ---        & TO     & ---      \\
urf6       & 15 & 2164  & 1215312   & 3.96    & 2215 & \bf1253000 & 69.89    & -3.10    & 2215 & 1244264    & 7.66   & -2.38    \\
dk27       & 15 & 23882 & 11254565  & 8598.83 & 8549 & 3998148    & 15257.82 & 64.48    & 4338 & \bf2141622 & 795.90 & 80.97    \\
C7552      & 20 & 257   & 180894    & 8.53    & 216  & \bf154448  & 27.52    & 14.62    & 367  & 274456     & 12.49  & -51.72   \\
bw         & 32 & 2585  & 3766784   & 2076.51 & ---  & ---        & TO       & ---      & ---  & ---        & TO     & ---      \\
 \hline \multicolumn{13}{p{.99\linewidth}}{%
   \vspace{0pt} The percentage improvements in the blocks \emph{Greedy} and
   \emph{Hamming} are with respect to the values for quantum costs ($q$) in the
   block \emph{BDD based}.  }
\end{tabularx}
\end{center}
\end{table}

\subsection{Comparison to Related Algorithms}
\label{sec:comp-relat-algor}
In this section, we compare our proposed algorithm to the original truth table
based variant from~\citep{VR:2008} and the QMDD based synthesis method
from~\citep{SWH+:2012}.  The truth table based approach has been re-implemented
in the RevKit program `\emph{young\_subgroup\_synthesis}'.  The considered
functions are taken from the LGSynth'93 benchmarks
(\url{www.cbl.ncsu.edu:16080/benchmarks/lgsynth93/}) and
from~\url{www.cs.uvic.ca/~dmaslov/}.  These functions are mainly irreversible
and provided in terms of their \emph{sum-of-product} representation saved as PLA
files.  We have used the embedding algorithm proposed in~\citep{SWG+:14} to
embed them as a reversible function that is represented by means of the binary
decision diagram of the characteristic function.  The time required for the
embedding is not accounted for in the reported run-times but can be obtained
from~\citep{SWG+:14}.
Since the embedding algorithm produces a partial function we 
first applied our proposed approach as this is the only approach out of the
three ones that supports partial functions.  Based on the resulting circuit we 
created a truth table representation using the RevKit program
`\emph{circuit\_to\_truth\_table}' which then was used as input to the truth
table based algorithm.  For the QMDD based algorithm we used the circuit as
input to construct the QMDD from which a different circuit is created.
Since the truth table based approach and the proposed approach create
single-target gates we used \texttt{exorcism}~\citep{MP:01} to translate
them into mixed-polarity multiple-control Toffoli gates.  The time required for
this translation is included in the overall reported run-times. The QMDD based
synthesis approach directly creates MPMCT gates.

Table~\ref{tbl:other-approaches} presents the results.  The first block of
columns denotes the name of the benchmark together with the number of circuit
lines~($n$).  For each algorithm a block lists the number of Toffoli
gates~($d$), the quantum costs in terms of $T$ gates in a Clifford+$T$ mapping
according to~\citep{AMMR:13}~($q$), and the run-time~(in seconds) required for
the whole synthesis~($t$).

It can be seen that our approach scales well and the performance
it is comparable to the QMDD based synthesis approach (fastest run-times are set
in bold face).  The truth table based approach was never the fastest one and
moreover, if the function has more than 15~variables, the truth table
based algorithm was not able to finish before the time-out.  It can be seen that
the run-time correlates much more to the number of variables for the truth table
based approach compared to the proposed approach.

In addition to the better run-times, the proposed approach also leads to better
results with respect to the number of Toffoli gates and quantum costs.  For the
latter, in the best case, we get an improvement of over~$83\%$ compared to the
original truth table based approach and over~$91\%$ compared to the QMDD based
approach.

\begin{remark}
  Due to Remark~\ref{rem:freedom} there must exist a choice of assignments such
  that the truth table based approach leads to the same results in terms of
  Toffoli costs as the proposed approach.  Hence, it seems that due to the
  symbolic decomposition described in Algorithm~D a better choice for the
  assignments is implicitly taken.
\end{remark}

\begin{remark}
  Most of the benchmarks used in Table~\ref{tbl:other-approaches} originally
  describe irreversible functions and were embedded as reversible functions.  We
  have used a different embedding approach~\citep{SWG+:14}.  This explains that
  in some cases a different number of lines and hence also to a different number
  of gates is obtained compared to the values given in~\citep{SWH+:2012}.
\end{remark}

\subsection{Evaluating Variable Ordering Heuristics}
\label{sec:eval-vari-order}
We have evaluated the heuristics for variable orderings as discussed in
Sect.~\ref{sec:ordering} for the same benchmarks as in the previous section and
listed the results in Table~\ref{tbl:orderings}.  Both the Greedy approach and
the approach based on the hamming distance can achieve significant improvement
in quantum costs.  The maximum improvement is over 92\% (\emph{co14}) for the
Greedy approach and over 80\% (\emph{dk27}) for the Hamming approach.

\begin{figure}[t]
\begin{displaymath}
  \begin{array}{>{\displaystyle}l>{\quad}r}
    F=\bigwedge_{i=1}^ny_i\leftrightarrow x_i & \text{(identity)} \\[7pt]
    F=\bigwedge_{i=1}^ny_i\oplus x_i & \text{(invert)} \\[5pt]
    F=\bigwedge_{i=1}^ny_i\leftrightarrow x_{(i+k)\operatorname{\mathrm{mod}}n + 1} & \text{(rotate)} \\[7pt]
    F=\bigwedge_{i=1}^ny_i\leftrightarrow
    \begin{cases}
      \bar x_i & \text{if $i$ is odd,} \\[7pt]
      x_{(i+2)\operatorname{\mathrm{mod}}n+1} & \text{if $i$ is even.}
    \end{cases} & \text{(invert or rotate)} \\[7pt]
    F=\bigwedge_{i=1}^{n/2}y_i\leftrightarrow x_i \land y_{2i}\leftrightarrow(x_i\oplus x_{2i}) & \text{(bitwise-xor)}
  \end{array}
\end{displaymath}
\caption{Functions for evaluating scalability}
\label{fig:scalability}
\end{figure}

The quantum cost is not always improved.  However, in case of the Greedy
approach the difference to the better solution is not too large.  In fact, in
these cases the gate count has actually decreased (which is also the cost
criteria in the implementation of the heuristics).  One can overcome this
problem by changing the cost criteria in the implementation which leads to a
higher run-time of the algorithm.  In case of the approach based on the hamming
distance, the quantum cost can also increase significantly (e.g.~\emph{C7752}).

The run-time for the Greedy approach is much higher whereas the run-time for the
Hamming approach is comparable to the original approach and usually correlates
with the achieved improvement.  If a higher improvement can be achieved, usually
less cycles need to be resolved and hence the run-time decreases.

\subsection{Evaluating Scalability}
In order to further evaluate the scalability of the proposed approach we have
created BDDs of characteristic functions that represent reversible functions
directly in memory.  Figure~\ref{fig:scalability} lists the functions that have
been used for this experiment and are parameterized by the number of lines.

Note that for \emph{invert or rotate} and \emph{bitwise-xor} the number of
lines~$n$ needs to be even.  We have implemented the experiment as the RevKit
test-case `\emph{rcbdd\_scalability}'.  The results are given by means of plots
in Fig.~\ref{fig:plots-dynamic}. The values of the $x$-axis and $y$-axis denote
the number of lines and run-time in seconds, respectively.
As can be seen functions with a large number of variables can be synthesized
with the proposed approach.  The run-time increases rapidly when the problem
instances get larger and the effect is more noticeable when the function is
complex.

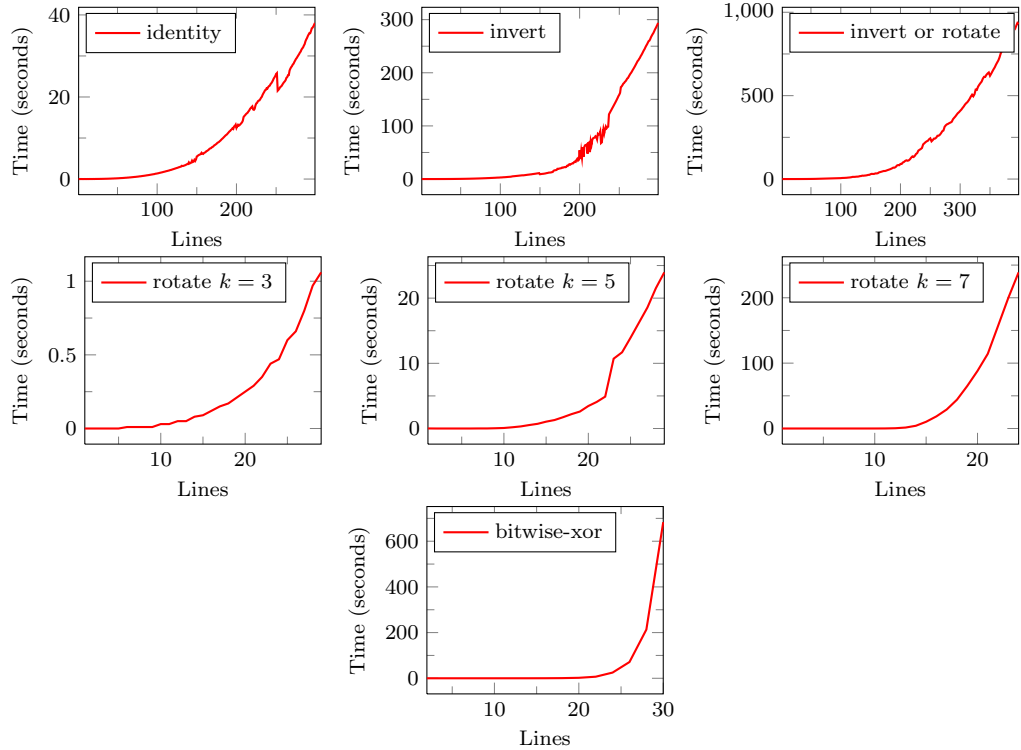
\begin{figure}[t]
\footnotesize
\pgfplotsset{%
  every axis/.append style={width=.35\linewidth,minor tick num=1,%
    enlarge x limits=false,legend pos=north west,xlabel=Lines,ylabel=Time (seconds),
    every axis y label/.style={at={(ticklabel cs:0.5)},anchor=near ticklabel,rotate=90},
    every axis x label/.style={at={(ticklabel cs:0.5)},anchor=near ticklabel}},
  plot/.style={red,thick},
  tick label style={font=\scriptsize}
}
\begin{center}
\begin{tikzpicture}
  \begin{axis}
    \addplot[plot] table [x=n,y=time] {table-create_identity.dat};
    \addlegendentry{identity}
  \end{axis}
\end{tikzpicture} \hfill
\begin{tikzpicture}
  \begin{axis}
    \addplot[plot] table [x=n,y=time] {table-create_invert.dat};
    \addlegendentry{invert}
  \end{axis}
\end{tikzpicture} \hfill
\begin{tikzpicture}
  \begin{axis}
    \addplot[plot] table[x=n,y=time] {table-create_invert_or_rotate.dat};
    \addlegendentry{invert or rotate}
  \end{axis}
\end{tikzpicture} \hfill
\begin{tikzpicture}
  \begin{axis}
    \addplot[plot] table[x=n,y=time] {table-create_rotate_3.dat};
    \addlegendentry{rotate $k=3$}
  \end{axis}
\end{tikzpicture} \hfill
\begin{tikzpicture}
  \begin{axis}
    \addplot[plot] table[x=n,y=time] {table-create_rotate_5.dat};
    \addlegendentry{rotate $k=5$}
  \end{axis}
\end{tikzpicture} \hfill
\begin{tikzpicture}
  \begin{axis}
    \addplot[plot] table[x=n,y=time] {table-create_rotate_7.dat};
    \addlegendentry{rotate $k=7$}
  \end{axis}
\end{tikzpicture}\hfill
\begin{tikzpicture}
  \begin{axis}
    \addplot[plot] table[x=n,y=time] {table-create_bitwise_xor.dat};
    \addlegendentry{bitwise-xor}
  \end{axis}
\end{tikzpicture}
\end{center}
\caption{Evaluating scalability}
\label{fig:plots-dynamic}
\end{figure}

\section{Conclusions}
\label{sec:conclusions}
In this paper we presented an algorithm for ancilla-free synthesis of large
reversible functions.  The reversible function is represented by its
characteristic function using binary decision diagrams.  This enables an
efficient symbolic function manipulation.  We formalized a decomposition
technique that has formerly been used in a truth table based synthesis approach
by means of co-factors on the characteristic function.  This enabled a synthesis
approach that is applicable to significantly larger functions.  Additionally, we
proposed heuristics to reduce gate cost and provide extensions to apply the
algorithm to partial functions.

An experimental evaluation demonstrates the applicability of the proposed
approach to large functions and also shows that the realized circuits lead to
smaller circuits compared to state-the-art synthesis approaches.

Most run-time is spent on resolving cycles.  We are currently investigating how
this process can be done more efficiently.  One way to overcome the efficiency
problem is to synthesize transpositions directly.  This would however lead to
circuits that are not structured in a V-shape and as a result a linear number of
single-target gates for their representation is no longer guaranteed.  An
increase in gate count and circuit cost is also expected.  Further future work
considers an efficient representation of single-target gates and also a direct
mapping of them to quantum gates.


\end{document}